\documentclass[twoside]{article}
\usepackage[accepted]{aistats2023}

\usepackage{hyperref}       
\usepackage{url}            

\usepackage{microtype}      
\usepackage{xcolor}  
\usepackage{wrapfig}
\usepackage{lipsum}
\usepackage{mathtools}
\usepackage{subcaption}
\usepackage{enumitem}
\usepackage{algorithm, algorithmic}
\newcommand{\ind}{\perp\!\!\!\!\perp} 
\newcommand{\notind}{\not\!\perp\!\!\!\!\perp}
\usepackage{longtable}
\usepackage{caption}

\usepackage{amsmath, amsfonts, amsthm, amssymb, amsbsy, graphicx, dsfont, mathtools, enumitem}
\usepackage{graphicx}
\usepackage{enumerate}
\usepackage{natbib}
\newtheorem{theorem}{Theorem}[section]

\newtheorem{lemma}{Lemma}[section]
\newtheorem{proposition}{Proposition}[section]
\newtheorem{assumption}{Assumption}[section]
\newtheorem{example}{Example}[section]
\theoremstyle{definition}
\newtheorem{definition}{Definition}[section]
 
\theoremstyle{remark}
\newtheorem*{remark}{Remark}

\begin{document}

\twocolumn[

\aistatstitle{A New Causal Decomposition Paradigm towards Health Equity}

\aistatsauthor{ Xinwei Sun* \And Xiangyu Zheng }

\aistatsaddress{ School of Data Science \\ Fudan University \And Department of Statistics, Guanghua School of Management \\ Peking University}

\aistatsauthor{Jim Weinstein*}

\aistatsaddress{Microsoft Corporation, Microsoft Research, Redmond, WA \\ Dartmouth, Tuck Business School \\ Northwestern University, Kellogg School of Business }
]

\begin{abstract}
Causal decomposition has provided a powerful tool to analyze health disparity problems by assessing the proportion of disparity caused by each mediator (the variable that mediates the effect of the exposure on the health outcome). However, most of these methods lack \emph{policy implications}, as they fail to account for all sources of disparities caused by the mediator. Besides, its identifiability needs to specify a set to be admissible to make the strong ignorability condition hold, which can be problematic as some variables in this set may induce new spurious features. To resolve these issues, under the framework of the structural causal model, we propose a new decomposition, dubbed as \emph{adjusted} and \emph{unadjusted} effects, which is able to include all types of disparity by adjusting each mediator's distribution from the disadvantaged group to the advantaged ones. Besides, by learning the \emph{maximal ancestral graph} and implementing \emph{causal discovery from heterogeneous data}, we can identify the admissible set, followed by an efficient algorithm for estimation. The theoretical correctness and efficacy of our method are demonstrated using a synthetic dataset and a common spine disease dataset.
\end{abstract}

\section{Introduction}
Health disparity/inequity, which is defined as the health outcome difference between socially advantaged and disadvantaged populations \citep{braveman2011health}, is a serious public health issue. To address this issue, a natural policy would be to identify all sources of disparities and develop policies to intervene and/or adjust for factors that mediate the effect of exposure (e.g., race or ethnicity) on the health outcome. To assess the effect of this policy, a natural question is: \emph{what is the amount of disparity reduction, after making a policy that removes all disparities through the mediator of interest?}

To answer this and related questions, many causal decomposition methods have been proposed, including natural indirect effects \cite{pearl2001direct, robins1992identifiability}, controlled mediated effects \cite{vanderweele2011controlled}, path-specific effects \cite{avin2005identifiability}, Oaxaca-Blinder decomposition \cite{oaxaca1973male}, among others, and have been applied to many health inequity scenarios \cite{ibfelt2013socioeconomic, hystad2013neighbourhood, jackson2018interpretation, blakely2018socioeconomic, mcguire2006implementing}. Although the common target is quantitatively attributing the disparity to the mediator via causal inference, they adopt different intervention strategies. Particularly, the natural indirect effect intervenes on the mediator to the value it would have taken had the exposure variable changed to the advantaged group and everything else had been the same; while the controlled mediated effect intervenes on other mediators at a fixed value. 

Different decomposition methods have different implications; however, when it comes to our goal, they either lack \emph{policy implications} or have issues with \emph{identifiability}. For instance, the natural indirect effect is more attributional than operational since it measures the effect under the original situation (\emph{i.e.}, the value of exogenous is unchanged) instead of the effect of a policy under a new situation. Besides, its identifiability requires that the exposure not affect the mediator-outcome confounders. In many scenarios, however, this assumption may not hold (if $R \to X$ in Fig.~\ref{fig:decompose}). Furthermore, when this assumption is violated, the natural indirect effect may fail to account for all sources of disparities (\emph{e.g.}, $M' \ind R$ when $X$ is an unobserved selection variable in Fig.~\ref{fig:decompose} (b)). Although \cite{jackson2020meaningful, vanderweele2014effect} can alleviate this problem, their identifiability is based on the strong ignorability condition, which means admissibility for some provided covariate sets, \emph{i.e.}, to deconfound between the mediator and the outcome. However, this condition may not hold when the provided covariate set includes some variables that perturb directed paths or induce new spurious features \cite{pearl2009causality}.

\textbf{Our Contributions.} To resolve these issues, we propose a new decomposition paradigm under the framework of the structural causal model (SCM), namely \emph{unadjusted and adjusted effects}, towards better policy implications and identifiability results, that account for social disparities. 

To be specific, we adjust each mediator's distribution from the disadvantaged group to the advantaged one, by generating mediators with an independent copy of exogenous variables. This has better policy implications since \textbf{i)} it is more aligned with the goal of assessing the effectiveness of a policy, and \textbf{ii)} can include all sources of disparities. Moreover, guided by the back-door criterion, we are able to identify the admissible set for appropriate attribution, by leveraging Markovian equivalence conditions of maximal ancestral graphs (MAGs) and exploring the heterogeneity among different social groups. With this guarantee, we propose an efficient algorithm to estimate the adjusted and unadjusted effects of social inequities. Theoretical correctness and utility are demonstrated by the consistent estimation results on a synthetic and robust spine-disease dataset.

\textbf{Organization.} In Sec.~\ref{sec.related}, we review existing causal decomposition methods. In Sec.~\ref{sec.pre}, we introduce some background knowledge regarding the structural causal model and causal discovery. In Sec.~\ref{sec.adjusted}, we propose a new causal decomposition framework: \emph{unadjusted} and \emph{adjusted} effects. We will show that compared to the natural indirect effect, our proposed adjusted effect is able to include all sources of disparity, thus having better policy implications. In Sec.~\ref{sec.iden}, we provide an identifiability analysis. Different from the strong ignorability condition that pre-specifies a covariate set to be admissible, we propose to identify the admissible set, via \textbf{i)} local causal discovery from heterogeneity data according to the exposure variable; and \textbf{ii)} exploring Markovian equivalence, which is more computationally efficient but requires additional conditions to determine equivalence. Guided by this analysis, we in Sec.~\ref{sec.estimate} introduce our estimation method. In Sec.~\ref{sec.experiment}, we evaluate our method on a synthetic dataset and a spine-disease dataset. Sec.~\ref{sec.conclusion} concludes the paper.

\section{Related Work}
\label{sec.related}

Existing methods for health inequity decompose the total amount of disparity into the proportion directly from exposure to the reported outcome, and the proportion indirectly through the mediated factor. Typical decomposition includes natural indirect effects, path-specific effects, controlled mediated effects, Oaxaca-Blinder decomposition, \emph{etc.} Perhaps the most familiar method is natural indirect (mediation)/direct effects \citep{pearl2001direct, vansteelandt2012natural, valeri2013mediation}, which have been applied to many health inequity problems \citep{ibfelt2013socioeconomic, hystad2013neighbourhood}. However, this effect is more attributional \citep{pearl2001direct} than operational since it assesses the change under the original situation, which may change at the time when the new policy is made. Besides, its identifiability requires no mediator-outcome confounder to be affected by the exposure, which may not be satisfied in health equity scenarios \citep{vanderweele2014effect}. Besides, when such confounders exist, the natural indirect effect may fail to account for all sources of disparities caused by the mediator of interest. Similarly, the controlled mediated effect, which assigns a fixed value to other mediated factors, is of less policy interest as it does not allow these factors to vary across individuals by default. The path-specific effect \citep{jackson2018interpretation, avin2005identifiability, gong2021path} measures the effect along some specific paths. As the intervention is implemented on the edge rather than the exposure or mediator variables, it has to integrate over the distribution of the exposure to obtain the path-specific effect; hence, it fails to measure the change after intervention (policy making). The Oaxaca-Blinder decomposition \cite{oaxaca1973male, quirk2006administrative} linearly relates the outcome to other covariates and is later endowed with causal meaning \citep{jackson2018decomposition} when the regression function refers to structural equations and the provided covariate set happens to deconfound the mediator and the outcome. Furthermore, the non-parametric form of this decomposition is proposed in \cite{jackson2018decomposition}. Other types of decomposition include \citep{jackson2020meaningful, vanderweele2014effect}. Similar to Oaxaca-Blinder decomposition, these methods also presume the admissibility of a provided covariate set (\emph{i.e.}, strong ignorability), which can be problematic since the provided covariate set may induce new spurious features, violating the strong ignorability condition. 

\textbf{In contrast}, our decomposition method, namely adjusted and unadjusted effects, can resolve the above defects in terms of \emph{policy implications} and \emph{identifiability}. Specifically, we define the adjusted effect as the difference between the advantaged and disadvantaged groups after equalizing the mediator's distribution, in the framework of structural causal models (SCMs) that serve an intuitive exhibition of the difference from other decomposition methods clearly. We show that our adjusted effect can well assess the disparity reduction after making an policy, by eliminating all sources of disparity. Besides, for identifiability, equipped with the back-door criterion in SCM, we establish the theory and the procedure to identify the admissible sets among covariates rather than presume by default that some pre-specified covariate sets are admissible. These results can extend to other estimators, such as \citep{jackson2020meaningful} which lets the mediator additionally condition on other variables.

\section{Preliminary}
\label{sec.pre}

\noindent \textbf{Structural Causal Model (SCM).} The SCM is defined as $\mathcal{M}:=\langle G,\mathcal{F}, \mathbf{U}, P(\mathbf{U}) \rangle$. Here, the $G:= (\mathbf{V},\mathbf{E})$ is a directed acyclic graph (DAG) that describes the causal structure over endogenous variables $\mathbf{V}$, where $\mathbf{E}$ denotes the edge set such that $X \to Y \in \mathbf{E}$ means a direct causal relation from $X$ to $Y$. We say $X$ causally influences $Y$ if there is a directed path from $X$ to $Y$. We refer to $\text{An}(V), \text{Pa}(V), \mathrm{De}(V), \text{Ch}(V)$ as ancestral, parent, descendant, and child variable sets of $V$. The $\mathbf{U}:=\{U_i\}$ denote exogenous variables, with $P(\mathbf{U})$ denoting its probability function. The $\mathcal{F} := \{f_{V_i}(Pa(i),U_i)\}_{V_i \in \mathbf{V}}$ denotes the set of structural equations. These structural equations are autonomous to each other, \emph{i.e.}, breaking one equation will not affect others. Under the \emph{Markovian} assumption that $\mathbf{U} := \{U_i\}$ are independent, we can factorize joint distribution into \emph{disentangled} factors \citep{scholkopf2021toward}, \emph{i.e.}, $P(\mathbf{V}) := \Pi_{V_i \in \mathbf{V}} P(V_i|Pa(i))$. These permit us to define the interventional distribution, \emph{i.e.}, "$P(\mathbf{O}|do(O_i=o_i))$", and study the (conditional) average causal effect (ACE): $\mathbb{E}(Y|do(T=1),X=x) - \mathbb{E}(Y|do(T=0),X=x)$. Graphically speaking, $do(O=o)$ means removing arrows into $O$ and setting its value to $o$. $G_{\overline{V}}$ and $G_{\underline{V}}$ respectively denote the graph with arrows going into and emanating from $V$ deleted. With the SCM, we further define the sub-model $\mathcal{M}_x$ as $\mathcal{M}$ with $X \gets f_X(Pa(X),U_x)$ replaced by $X=x$. The counterfactual $Y(X=x)$ is then defined as $Y(X=x)(\mathbf{U}=u):= Y_{\mathcal{M}_x}(\mathbf{U}=u)$.

The SCM translates the strong ignorability assumption, \emph{i.e.}, $\{Y(1),Y(0)\} \ind T|Z,X$ into a more intuitive graphical criteria \citep{galles2013testing} for $Z$ to be admissible: $P(Y|do(T=t),X) = \int P(Y|T=t,Z,X)P(Z|X)dZ$. Among these criteria, the back-door criterion is the most typical one, which refers to a subset that \textbf{i)} blocks each back-door path from $T$ to $Y$ and \textbf{ii)} meanwhile does not perturb directed paths or induce new spurious features. This criterion, especially \textbf{ii)}, informs us that the presumption of admissibility for some covariate sets may be problematic when it includes variables that violate \textbf{ii)}.

\noindent \textbf{Causal Discovery.} Equipped with the SCM that converts the strong ignorability condition into graphical criteria, the admissible sets can be easily detected via causal discovery. Typical methods for causal discovery include the PC algorithm, which first learns an undirected skeleton via conditional independence test and identifies as many directions as possible via v-structure, \emph{e.g.}, $V_i \to V_k \gets V_j$.

\noindent \textbf{Maximal Ancestral Graph.} Indeed, when there are unobserved variables, one can learn a \emph{maximal ancestral graph} (MAG) over a subset node \citep{zhang2005transformational}, which is a mixed graph that preserves the conditional independence relation in the sense that two variables are not adjacent if and only if there is no inducing paths\footnote{We leave the definition of inducing path to the supplementaries} between them. Specifically, the MAG may contain: 1) $X \to Y$, 2) $X \gets Y$ and 3) $X \leftrightarrow Y$, which respectively mean 1) $X \in \mathrm{An}(Y)$, 2) $Y \in \mathrm{An}(X)$ and 3) $X$ and $Y$ is correlated but $X \not\in \mathrm{An}(Y)$, $Y \not\in \mathrm{An}(X)$. We define $\circ$ that can refer to the head or tail of an arrow. \cite{zhang2005transformational, spirtes1996polynomial} provided conditions for two MAGs to be \emph{Markovian equivalent} (in Prop.~\ref{prop.mag-equi}), which can help to determine causal directions that differentiate among classes of MAGs. 
\begin{proposition}[Proposition 1 in \citep{spirtes1996polynomial}]
\label{prop.mag-equi}
Two MAGs are Markovian equivalent iff \textbf{i)} they have the same adjacencies; \textbf{ii)} they have the same unshielded colliders; \textbf{iii)} for any discriminating path of a vertex $V$, $V$ is a collider on the path in one graph iff it is a collider on the path in the other. Here, $B$ is an unshielded collider if in $(A,B,C)$, $A \circ\!\!\!\to B \gets\!\!\!\circ  C$. A discriminating path $u:=\langle X,W_1,...,W_K,V,Y \rangle$ means $u$ includes three edges, with $X$ not adjacent to $Y, V$ adjacent to $Y$, and any $W_k$ (for $k\in \{1,...,K\}$) is a collider and a parent of $Y$. 
\end{proposition}

\noindent \textbf{Identifying Directions from Heterogeneity.} Recently, \cite{huang2020causal} proposed to leverage domain index $C$ to identify directions (beyond v-structure) among variables with changed causal mechanisms, \emph{i.e.}, $P^c(V_i|Pa(i))$ changes across $C$. This can be achieved under a ``faithfulness" condition at distributional level, \emph{i.e.}, if $C \to V_i$, $C \to V_j$ and $V_i \ind V_j | C, X$ for some $X$, then $\{P(V_i|X,C=c\}$ is independent of $\{P(V_j|V_i, X,C=c)\}$ due to disentanglement of generating factors $\{P(V_i|Pa(i))\}$; while $\{P(V_j|X,C=c)\}$ and $\{P(V_i|V_j, X,C=c)\}$ are dependent. One can thus determine $V_i \to V_j$ (\emph{resp.} $V_j \to V_i$) if the Hilbert Schmidt Independence Criterion (HSIC) \citep{gretton2007kernel} $\Delta_{V_i \to V_j|X} < \alpha$ (\emph{resp.}, $\Delta_{V_j \to V_i|X} < \alpha$) for some significance level $\alpha > 0$.

\section{Adjusted and Unadjusted Effect}
\label{sec.adjusted}

\textbf{Problem Setup.} Our data consists of $\{r^{(i)}, o^{(i)},y^{(i)}\}_{i=1}^n$ $\sim_{i.i.d} P(R,\mathbf{O},Y)$, where $R$ denotes the exposure variable with $R=1$ (\emph{resp.} $R=0$) denoting the advantaged (\emph{resp.} disadvantaged) group, $Y$ denotes the health outcome, and $\mathbf{O}$ denotes the observed covariate set. Denote $Y(R=r)$ as the counterfactual. The health inequity means $\mathbb{E}(Y(R=1)) - \mathbb{E}(Y(R=0)) > 0$. Our goal is to assess the disparity reduction to measure the effect of making policy on some mediator variables $M \in \mathbf{M} \subset \mathbf{O}$. Here, we call $M$ a mediator if there exists a directed causal path from $R$ to $Y$ that goes through $M$.

\begin{figure*}[ht!]
 \centering
 \begin{subfigure}[t]{0.33\textwidth}
 \centering
        \includegraphics[width=1.3in]{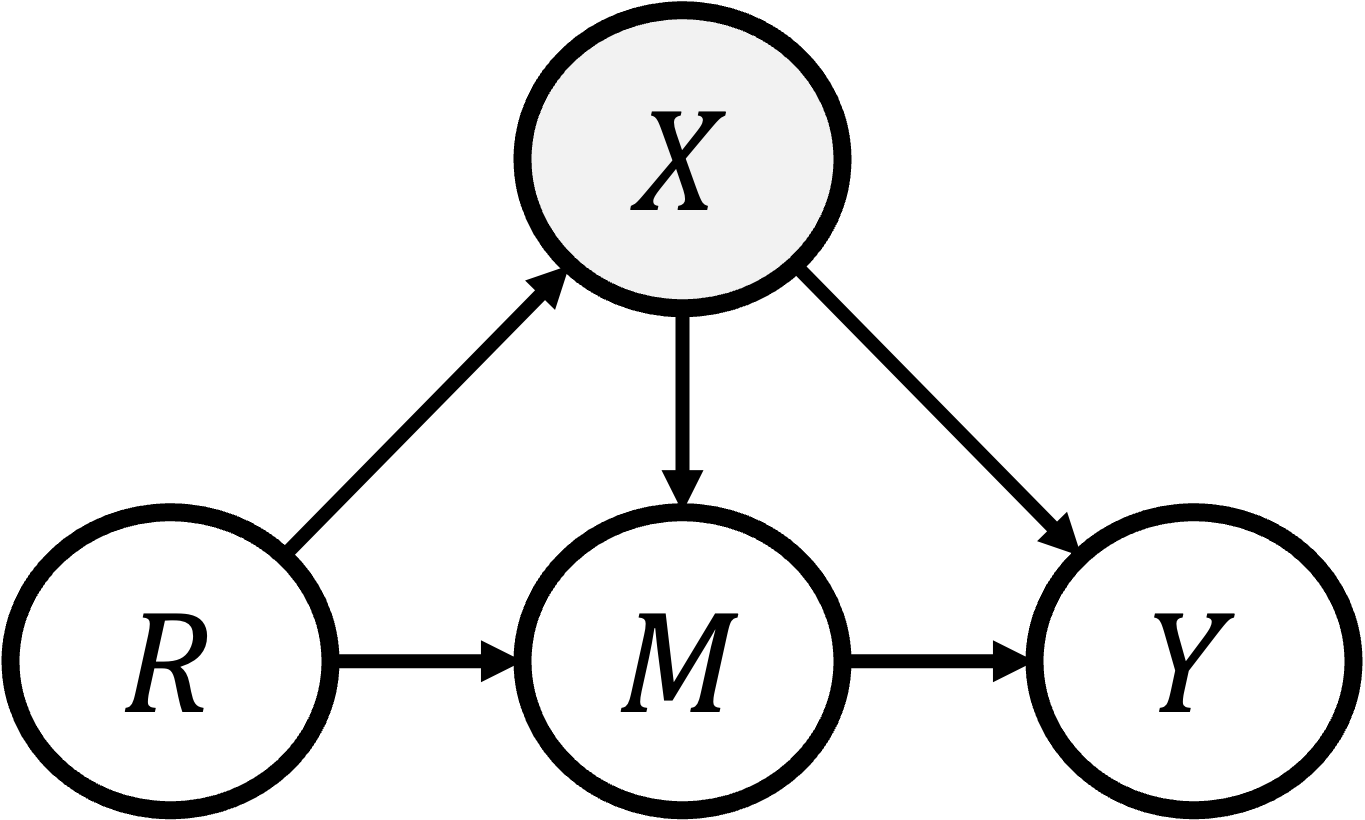}
        \caption{(a) Original causal graph}
        \label{fig:ground-truth-graph}
\end{subfigure}
 \begin{subfigure}[t]{0.33\textwidth}
 \centering
        \includegraphics[width=1.5in]{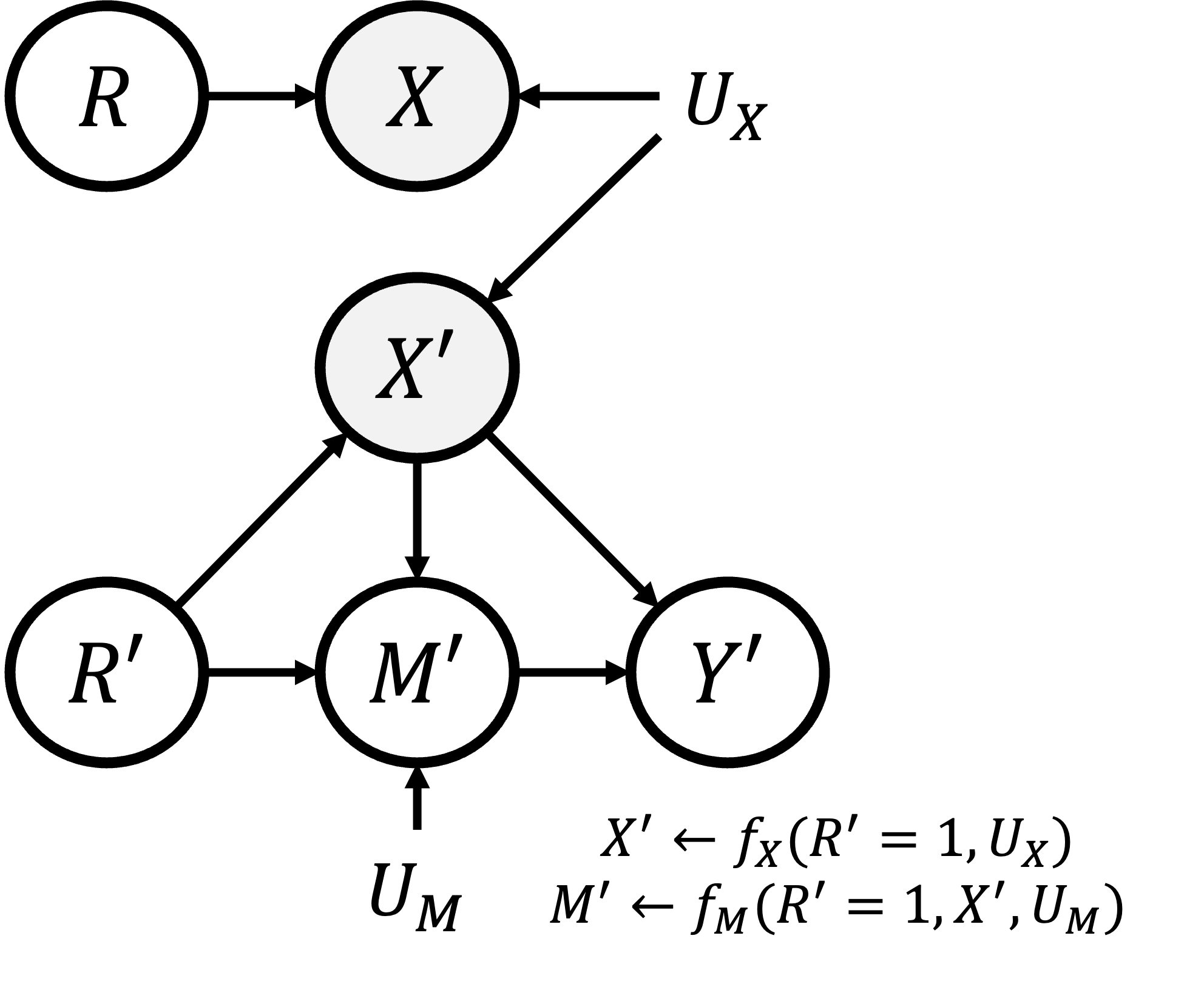}
        \caption{(b) Mediation $M(R'=1)(\mathbf{u})$}
        \label{fig:ground-mediation}
\end{subfigure}
 \begin{subfigure}[t]{0.33\textwidth}
 \centering
        \includegraphics[width=1.5in]{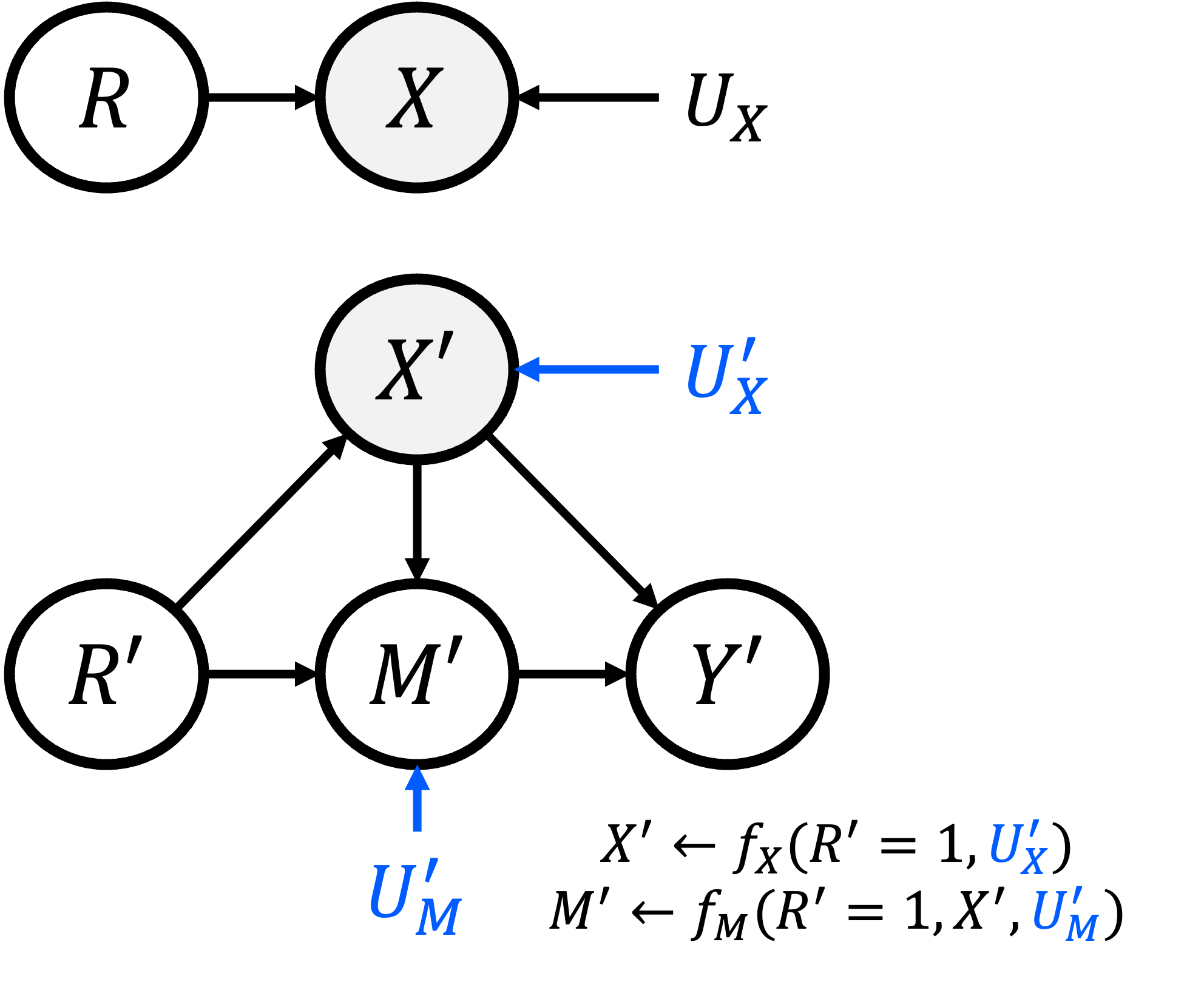}
        \caption{(c) Adjusted (Ours) $M_{R'=1}(\mathbf{u}')$}
        \label{fig:ground-adjusted}
\end{subfigure}
\caption{$R, M$ respectively denote exposure and mediator variables; $X$ is unobserved; $V'$ is a copy of $V$ after making policy, generated by another group $R' \neq R$. Unlike the \emph{Mediation Effect} with original exogenous $\mathbf{U}$; our \emph{Adjusted Effect} uses $\mathbf{U}' \sim_{i.i.d} \mathbf{U}$ to generate $M'$ (Eq.~\eqref{eq.adj-M}), which can effectively account for all sources of disparities.}
\label{fig:decompose}
\end{figure*}

In this section, we introduce a new decomposition method: the \emph{unadjusted} and \emph{adjusted} effects. Formally speaking, $\forall \thinspace M \in \mathbf{M}$, we define the adjusted effect $\delta_M$ and the unadjusted effect $\zeta_M$ as:
\begin{align}
    \delta_M & := \mathbb{E}(Y(R=0,G_{M|R=1})) - \mathbb{E}(Y(R=0)), \label{eq.adjusted} \\
    \zeta_M & := \mathbb{E}(Y(R=1)) - \mathbb{E}(Y(R=0,G_{M|R=1})), \label{eq.unadjusted}
\end{align}
where $G_{M|R=1}$ means to replace the original factor $P(M|Pa(M))$ with $P(M|R=1)$, \emph{i.e.}, assigning the mediator $M$ with the distribution of the advantaged group $M|R=1$. We use $V_i'$ to denote the copy of $V_i$ generated by the new policy from $R' \neq R$, in order to differentiate the original $V_i$. Then, correspondingly, structural equations for $Y(R=0,G_{M|R=1})$ are: 
\begin{align}
& R \gets 0, \label{eq.adj-R} \\
& V_i \gets f_i(Pa(V_i),u_i), \ \forall \thinspace V_i \in \mathbf{V} \backslash \mathrm{De}(M), \label{eq.adj-V} \\
& \mathbf{G_{M|R=1}}: \ \mathbf{M \gets M_{R=1}(\mathbf{u}')}, \ \mathbf{u}' \sim P(\mathbf{U}), \label{eq.adj-M} \\
& V_i \gets f_i(Pa(V_i)\backslash \mathbf{M}, \mathbf{M},u_i), \ \forall \thinspace V_i \in  \mathrm{Ch}(\mathbf{M}), \label{eq.adj-ch-m} \\
& V_i \gets f_i(Pa(V_i),u_i), \ \forall \thinspace V_i \in \mathrm{De}(\mathbf{M}) \backslash \mathrm{Ch}(\mathbf{M}), \label{eq.adj-de-m}
\end{align}
Here, Eq.~\ref{eq.adj-R},~\ref{eq.adj-V},~\eqref{eq.adj-ch-m},~\eqref{eq.adj-de-m} describe the original generating process under $\mathbf{u}$; while Eq.~\ref{eq.adj-M} describes the regenerating process of $M$ under the intervened model $M_{R=1}$ with $\mathbf{u'} \sim P(\mathbf{U})$ that is independent to $\mathbf{u}$. Finally, we regenerate $M$'s children and descendants using Eq.~\ref{eq.adj-ch-m},~\eqref{eq.adj-de-m}. Simply speaking, the adjusted effect replaces $M \gets f_M(R=0,Pa(M)\setminus R, u_M)$ in the original SCM with an intervened policy on $M$: $M_{R=1}(\mathbf{u}' \sim P(\mathbf{U}))$ that equalizes the distribution of $M$ between advantaged and disadvantaged groups.

Compared with the natural indirect effect with $Y(R=0,M(R=1))$, our adjusted effect generates $M$ using another exogenous $\mathbf{U}' \sim_{i.i.d} \mathbf{U}$, which has better policy implications in terms of measuring the effect of making policy in a new situation. Besides, it can effectively account for all sources of disparity. To illustrate, consider the example in Fig.~\ref{fig:decompose}:
\begin{example}
\label{exam:1}
Suppose $R, X \text{ (unobserved)}, M, Y$ in Fig.~\ref{fig:decompose} (a) respectively denote the race, retirement income, medical insurance, and health outcome. The disadvantaged group has a lower level of retirement income and thus a lower level of medical insurance. As illustrated in Fig.~\ref{fig:decompose} (b), the structural equations for natural indirect effect $Y(R=0,M(R=1))$ are:
\begin{align*} 
& \begin{cases}
  R \gets 0, \\
  X \gets f_X(R=0,U_X), \\
  \mathbf{X' \gets f_X(R'=1,U_X), M' \gets f_M(X',R'=1,U_M)},\\
  Y \gets f_Y(X,M',U_Y),
\end{cases} \\
\MoveEqLeft[-1]\text{SCM of $Y(R=0,M(R=1))$}
\end{align*}
in which $U_X$ is the confounder between $X$ and $X'$ to generate $M'$. When $X$ is unobserved, $M'$ may not eliminate all sources of disparities, \emph{i.e.}, $M' \notind R$, either when $X$ is the selection or latent variable. Specifically, when $X$ is taken as the selection variable, then as the collider between the path from $R =0$ to $M'$ (in Fig.~\ref{fig:decompose} (b)) it induces the correlation between $R = 0$ and $M'$, which means $M'$ still mixes the information of $R=0$. When $X$ is the latent variable, it can be mistakenly taken as exogenous of $M$, \emph{i.e.}, $X \in \mathbf{U}_M$ since the condition that whether there is another variable affecting $X$ is unknown. Since the exogenous is unchanged in natural indirect effect, $M' \gets f_M(R'=1,X,\mathbf{U}_M \backslash X)$ is generated with $X$ that is generated from $R=0$. In this case, the $M(R=1)(\mathbf{u})$ of the natural indirect effect fails to account for the disparity in the path of $R \to X \to M$. 
\begin{align*}
  & \begin{cases}
  R \gets 0, \\
  X \gets f_X(R=0,U_X),\\
  \mathbf{X' \gets f_X(R=1,\tilde{U}_X), M \gets f_M(X',R=1,U_M)},\\
  Y \gets f_Y(X,M,U_Y).
  \end{cases} \\
\MoveEqLeft[-1]\text{SCM of $Y(R=0,G_{M|R=1})$}
\end{align*}
In contrast, our adjusted effect with $Y(R=0,G_{M|R'=1})$ is able to remove all disparities sourced from medical insurance $M$, \emph{i.e.}, $M_{R'=1} \ind R$, since it generates $M'$ with independent exogenous $\mathbf{U}'$. 
\end{example}
In addition to policy implications, our adjusted effect requires weaker conditions for identifiability, compared to the natural indirect effect. Specifically, the natural indirect effect requires the deconfounding set between $M$ and $Y$ (given $R$) to not contain the descendants of $R$. This assumption may not hold in our health equity scenario, as shown by example~\ref{exam:1}. In contrast, our $Y(R=0,G_{M|R=1})$ only requires no existence of unobserved confounders between $M$ and $Y$, which is commonly made in the literature, including the natural indirect effect and other decomposition methods. 

Indeed, our definition is closely related to others in the literature. Specifically, when the linear equations refer to structural equations, our $Y(R=0,G_{M|R=1})$ degenerates to Oaxaca–Blinder decomposition \citep{oaxaca1973male, quirk2006administrative}. Besides, it has the same form as randomized intervention effects \citep{vanderweele2014effect, jackson2018decomposition}, if the covariate set provided in their works happens to be admissible. Further, our definition can be easily extended to \citep{jackson2020meaningful} that lets $M$ and $Y$ additionally condition on allowable variables. \emph{More importantly}, unlike these methods that assume without checking that the covariate set $\mathbf{O}$ happens to be admissible, our adjusted effect is endowed with a method that is provable to identify the admissible set among $\mathbf{O}$.

\section{Identifiability}
\label{sec.iden}

In this section, we discuss the identifiability property of $\delta_{M}$. We first give some basic assumptions. 

\begin{assumption}[SCM with no Unobserved Confounders]
\label{assum:scm}
  We assume that $\langle G, \mathcal{F}, \mathbf{U}, P(\mathbf{U}) \rangle$ is a structural causal model, with $G := (\mathbf{V},\mathbf{E})$ and $\mathbf{O} \cup \{Y,R\} \subset \mathbf{V}$. In addition, we assume for each $M \in \mathbf{M}$, there is no unobserved confounder between $M$ and $Y$. 
\end{assumption}

Note that we do not assume all variables in $\mathbf{V}$ are observable; rather, we only observe a subset of them $\{\mathbf{O},R,Y\} \subset \mathbf{V}$. The ``no unobserved confounder" condition, which does not additionally pre-specifies some covariate sets to be admissible, is weaker than the strong ignorability condition that is widely assumed in the literature of causal inference \citep{pearl2009causality, spirtes2000causation} and other causal decomposition works \cite{jackson2020meaningful, vanderweele2014effect}. Furthermore, we will show later that this assumption is necessary for the adjusted effect to be identifiable. 

In the following, we define the \emph{context variable}, which refers to the variable with no arrows going into them. 

\begin{definition}[Context Variable]
\label{def:context}
  We say $\mathcal{C} \subset \mathbf{O} \cup R$ is the set of \emph{context variables} if any $C \in \mathcal{C}$ is not causally affected by any variable in $\mathbf{V}$. 
\end{definition}
\begin{remark}
It naturally holds that $R$ is a context variable, \emph{i.e.}, $R \to \mathcal{C}$ since the exposure variable normally refers to the innate attributes. Similarly, other demographic variables such as gender, and ethnicity are also context variables. 
\end{remark}
Besides, as $Y$ refers to the outcome that is temporally reported later than others, it naturally holds that $Y$ does not causally influence other variables. Then we have {\small $\delta_M = \int yp(y|do(m),R=0)p(m|R=1)dydm - \int yp(y|R=0) dy$}, due to $(R \ind Y)_{G_{\overline{M},\underline{R}}}$, $(R \ind Y)_{G_{\underline{R}}}$, and $(R \ind M)_{G_{\underline{R}}}$. To identify $\delta_M$, all that is left is to identify $p(y|do(m),r)$. 

Next, we introduce Markovian and Faithfulness conditions. 
\begin{assumption}[Markovian and Faithfulness]
\label{assum:markov}
  The \emph{Markovian} means that the exogenous variables $\{U_i\}$ are independent. This implies $\forall$ disjoint sets $\mathbf{V}_i, \mathbf{V}_j, \mathbf{V}_{k}$, it holds that $\mathbf{V}_i \ind_d \mathbf{V}_j | \mathbf{V}_k \implies \mathbf{V}_i \ind \mathbf{V}_j | \mathbf{V}_k$. On the other hand, \emph{faithfulness} means $\mathbf{V}_i \ind \mathbf{V}_j | \mathbf{V}_k \implies \mathbf{V}_i \ind_d \mathbf{V}_j | \mathbf{V}_k$, where $\ind_d$ and $\ind$ respectively mean d-separation and probability independence.
\end{assumption}

This assumption, as was commonly made in the causal inference literature, allows us to implement conditional independence tests for causal discovery, which can provide a graphical criterion to identify $p(y|do(m),r)$. In the following, we show that it is almost necessary to assume no unobserved confounder in assump.~\ref{assum:scm} to identify $p(y|do(m),r)$. 

\begin{proposition}
\label{prop.back-door}
Suppose assumption~\ref{assum:markov} holds in the SCM $\langle G, \mathcal{F}, \mathbf{U}, P(\mathbf{U}) \rangle$, we have that the $p(y|do(m),r)$ is identifiable for all forms of $P(\mathbf{U})$ and $\mathcal{F}$, if and only if there exists an admissible set $\mathbf{B}_M$, \emph{i.e.}, $(M \ind Y|\mathbf{B}_M,R)_{G_{\underline{M}}}$. 
\end{proposition}

\begin{remark}
If there exists an unobserved confounder such that its path to $Y$ or $M$ is not blocked by $\mathbf{B}_M$, then $p(y|do(m),r)$ cannot be identified generally. Note that under the linear model, \citep{kasy2014instrumental, frolich2017direct} leverages instrumental variables for identifiability with unobserved confounders, however, assumed that $M(Z=z,u_M)$ increases with respect to the instrumental variable $Z$ for all individuals $u_M$. This condition may not hold in our health equity scenario. For example, if $Z:=R$ denotes race and $M$ denotes the income level, then it is unreasonable to assume that any person in the advantaged group has a higher income level than all persons in $R=0$. 
\end{remark}

The Prop.~\ref{prop.back-door} informs us to identify an admissible set, of which each $O_i$ is selected according to the type of the maximal ancestral graph (MAG) over $(O_i, M, Y)$. Specifically, we first list a broader family of types of MAGs in Fig.~\ref{fig:mag}, which are equivalent according to Prop.~\ref{prop.mag-equi}. We show that any set that contains variables belonging to type (a) and does not contain any variables in types (d,e,f) is admissible. To explain, note that the variable with the type (a) is sufficient to block each back-door path between $M$ and $Y$ as there is no unobserved confounder between $M$ and $Y$. Besides, the variable with types (b) and (c) does not perturb directed paths or induce new spurious features. Therefore, the inclusion of these variables can still make the set admissible. Types (d,e,f) are other members of the equivalence class. 
\begin{figure}[ht!]
    \centering
    \includegraphics[width=0.42\textwidth]{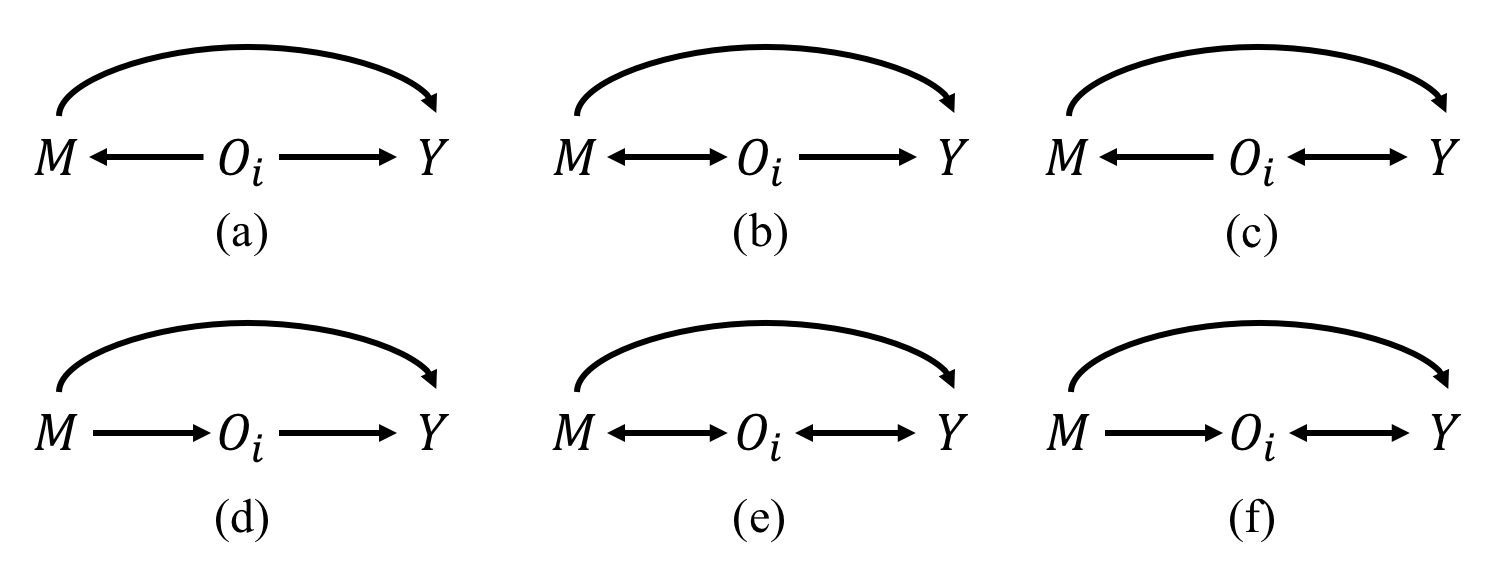}
    \caption{MAG of admissible sets: (a,b,c) and their equivalence classes: (d,e,f).}
    \label{fig:mag}
\end{figure}

Formally speaking, we define 
\begin{align*}
    \underline{\mathbf{B}}_M & :=\{O_i: \mathrm{MAG}(O_i,M,Y) \text{ is (a).}\}, \\
    \overline{\mathbf{B}}_M & :=\{O_i: \mathrm{MAG}(O_i,M,Y) \text{ belongs to (a,b,c).}\},
\end{align*}
and show that any set that contains $\underline{\mathbf{B}}_M$ and belongs to $\overline{\mathbf{B}}_M$ is an admissible set: 
\begin{theorem}
\label{thm.admiss}
Under assump.~\ref{assum:scm},~\ref{assum:markov}, any $\mathbf{B}_M$ such that $\underline{\mathbf{B}}_M \subset \mathbf{B}_M \subset \overline{\mathbf{B}}_M$ is admissible. 
\end{theorem}

According to Thm.~\ref{thm.admiss}, it is sufficient to identify any intermediate sets between $\underline{\mathbf{B}}_M$ and $\overline{\mathbf{B}}_M$. To achieve this goal, we first iteratively learn $\mathrm{MAG}(O_i, M, Y)$ for each $O_i \in \mathbf{O}$ via the FCI-JCI algorithm, which is provable to be sound and complete \cite{zhang2008causal} to identify the equivalence class of MAG. As members (d),(e),(f) are equivalent to types (a),(b),(c), we need to discriminate variables with types (a-c) from those with types (d-f). 

To achieve this goal, we propose to leverage the context variables in Def.~\ref{def:context}. Specifically, inspired by conditions (ii),(iii) in Prop.~\ref{prop.mag-equi}, we introduce the \emph{unshielded}-condition and the \emph{discriminate}-condition on the context variable: 

\begin{definition}
For the mediator $M$, we say $\mathcal{C}$ satisfies the \textbf{MAG-Equi} condition for $O_i$ and $M$, if either one of the following conditions holds: 
\begin{itemize}[noitemsep,topsep=0pt]
    \item \emph{Unshielded}-condition: there exists $C^i_M \in \mathcal{C}$ such that $\mathrm{MAG}(C^i_M, O_i, M)$ is $C^i_M \to O_i \circ\!-\!\circ M$; 
    \item \emph{Discriminate}-condition: there exists a pair of context variables $(C_M^{i,1},C_M^{i,2})$ such that \textbf{i)} $\mathrm{MAG}(C^{i,1}_M, O_i, M)$ is $C^{i,1}_M \to M \circ\!-\!\circ O_i$; and \textbf{ii)} $C^{i,2}_M$ is adjacent to $M$ and not adjacent to $Y$ in $\mathrm{MAG}(C^{i,2}_M, O_i, M, Y)$.
\end{itemize}
\end{definition}

As the name suggests, the \emph{unshielded}-condition and \emph{discriminate}-condition respectively exploit the unshielded collider criterion and the discriminating path criterion in Prop.~\ref{prop.mag-equi} to identify different classes of MAGs, with the assistance of context variable. In the following, we will show that if $\mathcal{C}$ satisfies the \textbf{MAG-Equi} condition, we can discriminate types (a,b,c) from (d,e,f) in Fig.~\ref{fig:mag}. 

\textbf{Identify $\overline{B}_M$ when the \textbf{MAG-Equi} condition holds.} The following lemma summarizes this result. 

\begin{lemma}
\label{lemma.detect-large}
Suppose assumptions~\ref{assum:scm},~\ref{assum:markov} hold and $\mathcal{C}$ satisfies the \textbf{MAG-Equi} condition. Then if $\mathcal{C}$ satisfies the \textbf{MAG-Equi} condition, we can discriminate types (a,b,c) from (d,e,f) in Fig.~\ref{fig:mag}. That is, we can determine whether $O_i \in \overline{\mathbf{B}}_M$. 
\end{lemma}
\begin{proof}[Proof-sketch]
Indeed, both \emph{unshielded}-condition and \emph{discriminate}-condition utilize ``unshielded colliders" criterion in Prop.~\ref{prop.mag-equi}. Specifically, {\small $C^i_M \to O_i \circ\!-\!\circ M$} eliminates types (d,e,f) with $M \circ\!\!\!\to O_i$. While for \emph{discriminate}-condition, the {\small $C^{i,1}_M \to M \circ\!-\!\circ O_i$} first eliminates types (d,f) with $M \gets\!\!\!\circ O_i$, in which $O_i$ is an unshielded collider. Further, the MAG structure over {\small $(C^{i,2}_M, M, O_i, Y)$} explores the discriminating path {\small $\langle C^{i,2}_M, M, O_i, Y \rangle$} for $O_i$ to eliminate type (e) with {\small $M \leftrightarrow O_i \leftrightarrow Y$}. Please refer to the appendix for details. 
\end{proof}
The \textbf{MAG-Equi} condition requires the existence of a context variable that does not simultaneously affect $O_i$ and $M$. This condition can hold for some $O_i$ when there are multiple context variables, such as gender, ethnicity, race, age, religious affiliation, and marital status. For instance, if $O_i,M,Y$ respectively denote education level, working style and treatment outcome, it is reasonable and hence probable for ``religious affiliation" to determine the working style $M$ but not affect the education level $O_i$ and $Y$. In this regard, the \emph{discriminate}-condition is satisfied when $C_M^{i,1}=C_M^{i,2}$ and denotes ``religious affiliation". Similarly, the \emph{unshielded}-the condition can be satisfied when $C_M^{i}$ denotes ``race/ethnicity" as it may not affect the working style $M$ given the education level $O_i$.

However, this \textbf{MAG-Equi} condition may uniformly hold for all $O_i$, especially when there are few context variables. In this case, lemma~\ref{lemma.detect-large} cannot be applied to determine whether $O_i \in \overline{\mathbf{B}}_M$. To address this problem, inspired by Thm.~\ref{thm.admiss}, we turn to discriminate type (a) from others, \emph{i.e.}, whether $O_i \in \underline{\mathbf{B}}_M$, by leveraging heterogeneity induced by the exposure variable $R$ to identify the causal directions $O_i \to M, Y$ that characterizes type (a).

\textbf{Identify $\underline{\mathbf{B}}_M$ when the MAG-Equi condition is violated.} We split data into multiple domains according to the value of $R$. The heterogeneity among these domains can help us to determine causal directions by exploiting changed causal modules \cite{huang2020causal}. This requires the following assumption, namely \emph{faithfulness} at the distributional level, which was firstly proposed by \cite{huang2020causal}. 

\begin{assumption}[\citep{huang2020causal}]
\label{assum:faith}
If $O_i \to O_j$ and at least $R \to O_i$ or $R \to O_j$ holds, then {\small $\{P^r(O_i|O_j,\mathbf{X})\}_r$} is dependent to {\small $\{P^r(O_j|\mathbf{X})\}_r$}, where $\mathbf{X}$ is the minimal deconfounding set \footnote{$X$ is a deconfounding set between $O_i$ and $O_j$ if $O_i \ind O_j|X$ and $X \cap (\mathrm{De}(O_i) \cup \mathrm{De}(O_j)) = \emptyset$.}. 
\end{assumption}

Under this assumption, we have the following result: 
\begin{lemma}
\label{lemma.detect-small}
Under assump.~\ref{assum:scm}-\ref{assum:faith}, we have {\small $O_i \in \underline{\mathbf{B}}_M$} iff there exists $\mathbf{X}^i_M$ such that \textbf{i)} {\small $\{P(M,Y|O_i,\mathbf{X}^i_M,r)\}_r$} and {\small $\{P(O_i|\mathbf{X}^i_M,r)\}_r$} are independent; \textbf{ii)} {\small $\{P(M|\mathbf{X}^i_M,r)\}_r$} and {\small $\{P(O_i|M,\mathbf{X}^i_M,r)\}_r$} are dependent; and \textbf{iii)} {\small $\{P(Y|\mathbf{X}^i_M,r)\}_r$} and {\small $\{P(O_i|Y,\mathbf{X}^i_M,r)\}_r$} are dependent.
\end{lemma}

To implement these independence tests, we correspondingly examine whether {\small $\Delta_{O_i,Y \gets M|\mathbf{X}_M^i} < \alpha$},  {\small $\Delta_{O_i \to M|\mathbf{X}_M^i} \geq \alpha$} and {\small $\Delta_{Y \to M|\mathbf{X}_M^i} \geq \alpha$} with some $\alpha > 0$ to respectively determine whether \textbf{i)}, \textbf{ii)} and \textbf{iii)} hold, derived from Hilbert-Schmidt Independence Criterion (HSIC) norm \cite{gretton2007kernel} to measure the dependency between two sets of distributions \cite{huang2020causal}. Compared to lemma~\ref{lemma.detect-large}, the lemma~\ref{lemma.detect-small} does not require the \textbf{MAG-Equi} condition, however, is more computationally expensive as it requires searching over all subsets of all $O_i$ with its type of MAGs belonging to (a-f) in Fig.~\ref{fig:mag}, \emph{i.e.}, {\small $\mathbf{B}^{\mathrm{all}}_M:=\{O_i: \mathrm{MAG}(O_i,M,Y) \text{ belongs to (a-f).}\}$}. We summarize lemma~\ref{lemma.detect-large},~\ref{lemma.detect-small} into the following theorem.

\begin{theorem}[Main Theorem]
\label{thm.iden}
Under assump.~\ref{assum:scm}-\ref{assum:faith}, there exists an identifiable admissible set $\mathbf{B}_M$ with $\underline{\mathbf{B}}_M \subset \mathbf{B}_M \subset \overline{\mathbf{B}}_M$. Particularly, $\mathbf{B}_M$ respectively degenerates to $\overline{\mathbf{B}}_M$ and $\underline{\mathbf{B}}_M$ if the \textbf{MAG-Equi} condition in lemma~\ref{lemma.detect-large} holds and does not hold for all $O_i \in \mathbf{O}$. 
\end{theorem}

\begin{remark}
This analysis can easily extend to the form in \citep{jackson2020meaningful} that additionally conditions on some ``allowable variables". As these variables are assumed to not be affected by $R$ and $M$, conditioning on these variables will not induce new spurious features or block directed paths. 
\end{remark}

With Thm.~\ref{thm.iden}, we can identify the adjusted effect $\delta_M$ and unadjusted effect $\zeta_M$ via $\mathbf{B}_M$ as:
{\small
\begin{align*}
    \delta_M & = \mathbb{E}_{m|R=1}\left[\mathbb{E}_{\mathbf{b}_M|R=0}\left[\mathbb{E}[Y|m,\mathbf{b}_M,R=0]\right]\right] - \mathbb{E}^0[Y], \\ 
    \zeta_M & = \mathbb{E}^1[Y] - \mathbb{E}_{m|R=1}\left[\mathbb{E}_{\mathbf{b}_M|R=0}\left[\mathbb{E}[Y|m,\mathbf{b}_M,R=0]\right]\right], 
\end{align*}
}
where $\mathbb{E}^r[Y]$ denotes $\mathbb{E}[Y|R=r]$ for $r=0,1$.

\section{Estimation}
\label{sec.estimate}

Algorithm~\ref{alg.estimation} summarizes the whole procedure to calculate $\delta_M, \zeta_M$ for each $M$. Roughly speaking, it contains \textbf{i)} detect all mediators $\mathbf{M}$; then for each $M \in \mathbf{M}$, \textbf{ii)} identify $\mathbf{B}_M$ and \textbf{iii)} estimate $\delta_M, \zeta_M$ by estimating $\mathbb{E}(Y|m,\mathbf{b}_{M},R=0)$ and $p(\mathbf{b}_{M}|R=r)$. 

\begin{algorithm}[ht!]
\caption{Estimate $\delta_M, \zeta_M$ for $M \in \mathbf{M}$.}
\label{alg.estimation}
\begin{flushleft}
\textbf{INPUT:} $\{r^{(i)}, o^{(i)},y^{(i)}\}_{i=1}^n$ and $\alpha > 0$. \\
\textbf{OUTPUT:} $\hat{\delta}_M$ for each $M \in \mathbf{M}$. 
\end{flushleft}
\begin{algorithmic}[1]
	\STATE Detect {\small $\mathbf{M}: = \{O_i|O_i \not\perp R,Y, O_i \not\perp R|Y, O_i \not\perp Y|R\}$}.
	\STATE For each $M \in \mathbf{M}$, 
	\STATE \quad Initialize {\small $\mathbf{B}_M=\emptyset$}, {\small $\mathbf{B}^{\mathrm{all}}_M = \emptyset$}. Denote $\mathbf{O}_M:=\mathbf{O} \backslash M$.
	\STATE \quad Run FCI-JCI to obtain {\small $\{\mathrm{MAG}(O_i,M,Y)\}_{O_i \in \mathbf{O}_M}$}.
	\STATE \quad Detect $\mathbf{B}^{\mathrm{all}}_M$ with types (a-f) in Fig.~\ref{fig:mag}. 
	\STATE \quad For each $O_i \in \mathbf{B}^{\mathrm{all}}_M$, 
	\STATE \quad \quad $\forall C_i \in \mathcal{C}$, run FCI-JCI to obtain {\small $\mathrm{MAG}(C_i,O_i,M)$} and {\small $\mathrm{MAG}(C_i,M,O_i,Y)$}. 
	\STATE \quad \quad If \textbf{MAG-Equi} holds, add $\mathbf{B}_M=\mathbf{B}_M \cup O_i$.
	\STATE \quad \quad Else, search over {\small $\mathbf{S} \subset \mathbf{B}^{\mathrm{all}}_{M}$} such that {\small $\Delta_{M,Y \gets O_i|\mathbf{S}} < \alpha$}, {\small $\Delta_{O_i \gets M|\mathbf{S}} > \alpha$} and {\small $\Delta_{Y \gets M|\mathbf{S}} > \alpha$}, add $\mathbf{B}_M = \mathbf{B}_M \cup O_i$; else continue.
	\STATE \quad Estimate {\small $\mathbb{E}(Y|m,\mathbf{b}_{M},R=0)$}, {\small $\mathbb{E}(Y|R=r)$} and {\small $p(\mathbf{b}_{M}|R=r)$} for $r=0,1$.
	\STATE \quad Estimate $\hat{\delta}_M$ and $\hat{\zeta}_M$. 
\end{algorithmic}
\end{algorithm}

For $\mathbb{E}(Y|m,\mathbf{b}_{M},R=0)$, we can choose several off-the-shelf predictors such as tree-based or generalized linear models for fitting. To estimate $p(\mathbf{b}_{M}|R=r)$ with $\mathbf{B}_M = \{B_i\}_{i=1}^s$, we write $p(\mathbf{b}_{M}|R=r) = \Pi_{i=1}^s p(b_i|b_{1:i-1},R=r)$ and implement kernel density estimation (KDE) to estimate $p(b_i|b_{1:i-1},R=r)$ if $b_i$ is continuous and use soft-max if $b_i$ is discrete. If $s$ is large, we use Markov chain Monte Carlo to generate data from {\small $p(\mathbf{b}_{M}|R=r)$}. 

\textbf{Computational Complexity.} Under the worst case, \emph{i.e.}, the \textbf{MAG-Equi} condition is violated for all $O_i$, one has to search over $S \subset \mathbf{B}^{\mathrm{all}}_M$. If we follow an ascending order in terms of $|S|$, the overall complexity to estimate $\delta_M,\zeta_M$ for each $M$ is with order $(p*|\mathcal{C}|)^3(\mathbf{B}^{\mathrm{all}}_M)^{|\underline{\mathbf{B}}_M|}$ with $p:=|\mathbf{O}|$, where $p^3$, $|\mathcal{C}|^3$ are respectively spent for calculating {\small $\{\mathrm{MAG}(O_i,M,Y)\}_{O_i}$} and {\small $\{\mathrm{MAG}(C_i,O_i,M)\}_{C_i}$, $\{\mathrm{MAG}(C_i,M,O_i,Y)\}_{C_i}$}. When the \textbf{MAG-Equi} condition holds for all $O_i$, the complexity is only with order $(p*|\mathcal{C}|)^3$.

\section{Experiment}
\label{sec.experiment}

In this section, we evaluate our method on a synthetic dataset and a spine disease dataset, namely \emph{Spine Patient Outcomes Research Trial} (SPORT) \cite{birkmeyer2002design, weinstein2008surgical, weinstein2006surgical1, weinstein2006surgical2, pearson2012should}.

\subsection{Simulation}
\label{sec.sim}

\noindent \textbf{Data Generation.} We follow the causal graph in Fig.~\ref{fig:graph-sim} to generate data via: {\small $R \sim \emph{Bern}(1,0.5)$}; {\small $X \gets f_X(\varepsilon_x)$ ($\varepsilon_x \sim \mathcal{N}(\mu_x,\sigma_x^2)$)}; {\small $M_1 \gets f_{M_1}(x,r) := \alpha_0 + \alpha_RR + \alpha_XX + \varepsilon_{M_1}$; $M_2 \gets f_{M_2}(x,r,m_1) := \beta_0 + \beta_1M_1 + \beta_RR + \beta_XX + \varepsilon_{M_2}$; $Y \gets f_Y(x,r,m_1,m_2):= \rho_0 + \rho_1M_1 + \rho_2M_2 + \rho_XX + \rho_RR + \varepsilon_Y$}. Our goal is to estimate adjusted and unadjusted effects for $M_1$, $M_2$, and $\{M_1,M_2\}$, respectively. 

\begin{figure}[ht!]
\centering
    \includegraphics[width=0.22\textwidth]{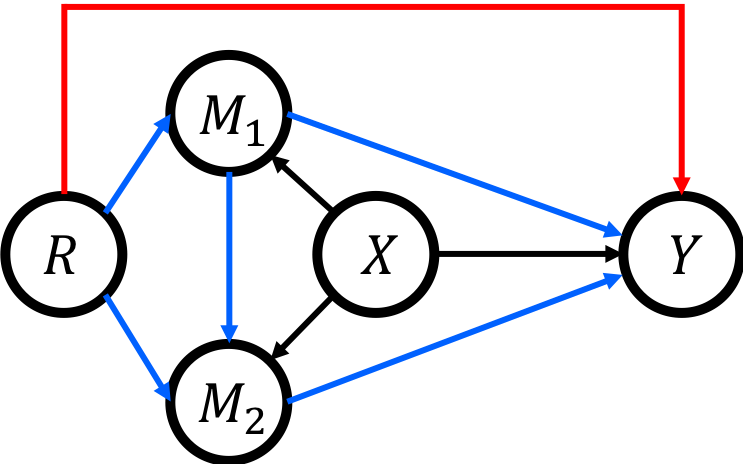} 
    \caption{Causal graph to generate data.}
    \label{fig:graph-sim}
\end{figure} 

Set {\small $(\alpha_0,\alpha_R,\alpha_X):= (0,2,3), (\beta_0, \beta_1, \beta_R, \beta_X) := (0,2,3,4)$} and {\small $(\rho_0,\rho_1,\rho_2,\rho_R,\rho_X)$} {\small $:= (0,2,3,4,5)$}, we have {\small $\delta_{M_1} = \rho_1 \mathbb{E}(M_1(1) - M_1(0)) + \rho_2\beta_1 \mathbb{E}(M_1(1)-M_1(0)) = (\rho_1 + \rho_2\beta_1)\alpha_1 = 16$}; {\small $\delta_{M_2} = \rho_2 \left( \beta_1\mathbb{E}(M_1(1) - M_1(0)) + \beta_2 \right) = \rho_2 (\beta_1\alpha_1 + \beta_2) = 21$}; {\small $\delta_{M_1,M_2} = (\rho_1 + \rho_2\beta_1) \mathbb{E}(M_1(1) - M_1(0)) + \rho_2 \beta_2 = (\rho_1 + \rho_2\beta_1)\alpha_1 + \rho_2 \beta_2 = 25$}; {\small $\zeta_{M_1} = \mathbb{E}(Y(1)) - \mathbb{E}(Y(0)) -  \delta_{M_1}= 14$; $\zeta_{M_2} = 9$} and {\small $\zeta_{M_1,M_2} = 5$}. 
\begin{figure*}[ht!]
    \centering
    \includegraphics[width=0.8\textwidth]{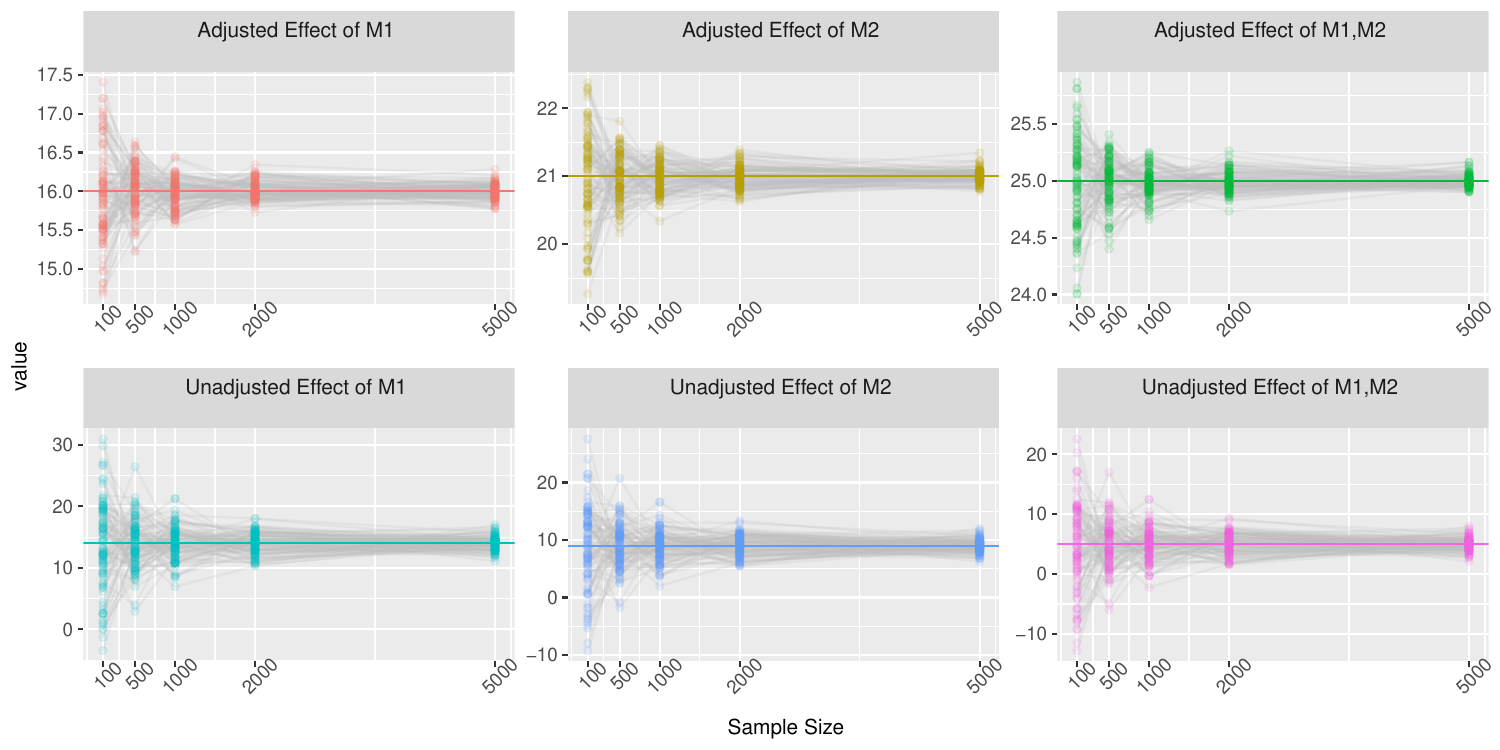}
    \caption{Estimation Results on Synthetic Dataset.}
    \label{fig:sim-result}
\end{figure*}

\begin{figure*}[ht!]
    \centering
    \includegraphics[width=6.0in]{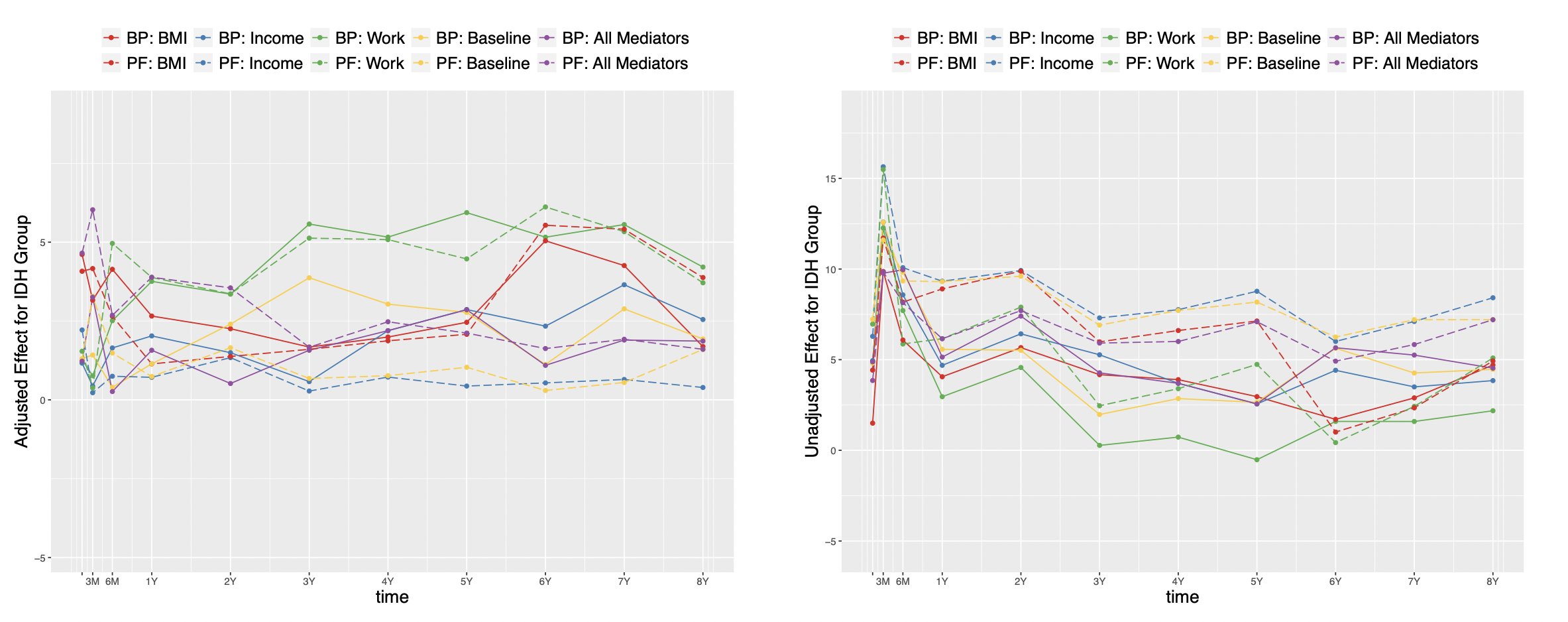}
    \caption{Adjusted/Unadjusted Effects on IDH data.}
    \label{fig:IDH}
\end{figure*}

\noindent \textbf{Implementation.} We follow algorithm~\ref{alg.estimation} wih $\alpha = 0.05$ and linear model to fit $\mathbb{E}(Y|m,\mathbf{b}_M,R=0)$. 

\noindent \textbf{Results.} Fig.~\ref{fig:sim-result} shows that our estimations can asymptotically approximate the ground-truth value as $n$ grows.

\subsection{Spine Patient Outcomes Research Trial (SPORT)}
\textbf{Dataset.} The Spine Patient Outcomes Research Trial (SPORT) dataset was designed to investigate the effectiveness of spine surgery for the three most common reasons for low back pain (LBP) surgical disease. Low back pain-related conditions remain one of the most controversial diseases, as to, whether surgery is more effective than non-surgical treatments \cite{birkmeyer2002design}. As reported in \cite{vos2017global}, LBP is a top-5 leading cause of disability globally, which has brought heavy burdens economically and socially to health systems globally. The SPORT study enrolled 2,505 patients in the United States from March 2000 to February 2005. There were three groups with associated LBP: 1) intervertebral disc herniation (IDH), 2) spinal stenosis (SPS), and 3) degenerative spondylolisthesis (DS). The race fell into three main groups: Asian Americans, African Americans, and White–Hispanic and non-Hispanic Americans. To investigate the long-term surgical effects (as LBP is a chronic disease), SPORT studied outcomes using validated measures. Primary outcomes were (i.e., $Y$) bodily pain (BP) and physical functioning (PF). Both are scaled from 0 to 100 (the higher, the better). Longitudinal data were collected at various intervals over 8 years after implementing surgical and non-surgical treatments (initially 6 weeks, 12 weeks, 6 months, 1 year, and yearly thereafter). For each patient, SPORT recorded the baseline demographic characteristics, social-economic factors such as income, working style, education years, and other commodities such as hypertension, diabetes, weight, smoking history, \emph{etc.} \cite{weinstein2008surgical, birkmeyer2002design}.

\textbf{Implementation.} We observe large differences between the white and non-white groups (which are mainly Asian and African Americans) in terms of BP and PF scores, as shown in Fig.~\ref{fig:total} in the appendix. We follow algorithm~\ref{alg.estimation} with $\alpha=0.05$ and use a random forest predictor with $50$ estimators to estimate $\mathbb{E}(Y|\mathbf{B}_M,M,R)$. The detected mediator set $\mathbf{M}$ contains four features: income and working styles (social economics status) Body mass index (BMI) and the outcome score at baseline (health status). We then calculate the adjusted effects for the white and non-white groups ($R = 1$ means white) for DS, SPS, and IDH with $n=601$, $634$, and $1,195$ patients; and $p=18$ covariates. More implementation details can be found in the appendix.

\textbf{Results Analysis.} Due to space limits, we only show the results on the IDH disease and placed other results in the appendix. As shown in Fig.~\ref{fig:IDH}, both adjusted and unadjusted effects are non-ignorable, which suggests more compensations for minorities (income, working status), healthier daily habits (BMI), and more personalized health care (the outcome at baselines) such as better medications. Particularly, for the working style and the income which are more of practical interest for policy-making, if we assign them to the advantaged group’s distribution, the disparity in both BP and PF scores can be decreased by a noticeable margin. To explain, we observed that the proportion of the non-white group with no or low-level income is 5.4\% more than that of the white patient subgroup and, that the proportion of employment of the black patient group was nearly 18\% less than that of the white subgroup. This suggests the possibility that alleviating disparities in outcomes can be achieved by focusing on one’s employment status and income. Besides, it is interesting to note that the BMI has significant adjustment effects, which help explain the disparities in terms of outcomes after surgical treatments. Finally, one can observe that making policies on all mediators (purple curve) does not necessarily have better effects, since the formula of $\delta_M$ does not necessarily increase as $|M|$ increases. To illustrate, consider the group $R=0$ that does not have good living or dietary habits, adjusting all factors may let them treat their disease lightly (e.g., choosing non-surgical treatment), which may in turn result in worse outcomes.

\section{Conclusion}
\label{sec.conclusion}
We define the adjusted and unadjusted effects under the structural-causal model framework. Compared to existing methods, our method has better policy implications, by eliminating all sources of disparity to measure the effect of a policy on the mediator of interest. Additionally, unlike traditional methods that rely on the strong ignorability condition for identifiability, we can identify the admissible set via causal discovery. Our methods are efficient and easy to implement. In the SPORT dataset, we identified important, interpretable socioeconomic variables that significantly impact disparities in terms of short and long-term treatment effects. For limitations, our methods require no unobserved confounder conditions. That said we plan to explore the ‘relaxation’ of this condition for future work.

\newpage
\bibliographystyle{IEEEtran} 
\bibliography{reference} 

\begin{thebibliography}{36}
\providecommand{\natexlab}[1]{#1}
\providecommand{\url}[1]{\texttt{#1}}
\expandafter\ifx\csname urlstyle\endcsname\relax
  \providecommand{\doi}[1]{doi: #1}\else
  \providecommand{\doi}{doi: \begingroup \urlstyle{rm}\Url}\fi

\bibitem[Avin et~al.(2005)Avin, Shpitser, and Pearl]{avin2005identifiability}
C.~Avin, I.~Shpitser, and J.~Pearl.
\newblock Identifiability of path-specific effects.
\newblock In \emph{Proceedings of the 19th International Joint Conference on
  Artificial Intelligence}, IJCAI'05, page 357–363, San Francisco, CA, USA,
  2005. Morgan Kaufmann Publishers Inc.

\bibitem[Birkmeyer et~al.(2002)Birkmeyer, Weinstein, Tosteson, Tosteson,
  Skinner, Lurie, Deyo, and Wennberg]{birkmeyer2002design}
N.~J. Birkmeyer, J.~N. Weinstein, A.~N. Tosteson, T.~D. Tosteson, J.~S.
  Skinner, J.~D. Lurie, R.~Deyo, and J.~E. Wennberg.
\newblock Design of the spine patient outcomes research trial (sport).
\newblock \emph{Spine}, 27\penalty0 (12):\penalty0 1361, 2002.

\bibitem[Blakely et~al.(2018)Blakely, Disney, Valeri, Atkinson, Teng, Wilson,
  and Gurrin]{blakely2018socioeconomic}
T.~Blakely, G.~Disney, L.~Valeri, J.~Atkinson, A.~Teng, N.~Wilson, and
  L.~Gurrin.
\newblock Socioeconomic and tobacco mediation of ethnic inequalities in
  mortality over time: repeated census-mortality cohort studies, 1981 to 2011.
\newblock \emph{Epidemiology (Cambridge, Mass.)}, 29\penalty0 (4):\penalty0
  506, 2018.

\bibitem[Braveman et~al.(2011)Braveman, Kumanyika, Fielding, LaVeist, Borrell,
  Manderscheid, and Troutman]{braveman2011health}
P.~A. Braveman, S.~Kumanyika, J.~Fielding, T.~LaVeist, L.~N. Borrell,
  R.~Manderscheid, and A.~Troutman.
\newblock Health disparities and health equity: the issue is justice.
\newblock \emph{American journal of public health}, 101\penalty0 (S1):\penalty0
  S149--S155, 2011.

\bibitem[Fr{\"o}lich and Huber(2017)]{frolich2017direct}
M.~Fr{\"o}lich and M.~Huber.
\newblock Direct and indirect treatment effects--causal chains and mediation
  analysis with instrumental variables.
\newblock \emph{Journal of the Royal Statistical Society: Series B (Statistical
  Methodology)}, 79\penalty0 (5):\penalty0 1645--1666, 2017.

\bibitem[Galles and Pearl(2013)]{galles2013testing}
D.~Galles and J.~Pearl.
\newblock Testing identifiability of causal effects.
\newblock \emph{arXiv preprint arXiv:1302.4948}, 2013.

\bibitem[Gong and Zhu(2021)]{gong2021path}
H.~Gong and K.~Zhu.
\newblock Path-specific effects based on information accounts of causality.
\newblock \emph{arXiv preprint arXiv:2106.03178}, 2021.

\bibitem[Gretton et~al.(2007)Gretton, Fukumizu, Teo, Song, Sch{\"o}lkopf, and
  Smola]{gretton2007kernel}
A.~Gretton, K.~Fukumizu, C.~Teo, L.~Song, B.~Sch{\"o}lkopf, and A.~Smola.
\newblock A kernel statistical test of independence.
\newblock \emph{Advances in neural information processing systems}, 20, 2007.

\bibitem[Huang et~al.(2020)Huang, Zhang, Zhang, Ramsey, Sanchez-Romero,
  Glymour, and Sch{\"o}lkopf]{huang2020causal}
B.~Huang, K.~Zhang, J.~Zhang, J.~D. Ramsey, R.~Sanchez-Romero, C.~Glymour, and
  B.~Sch{\"o}lkopf.
\newblock Causal discovery from heterogeneous/nonstationary data.
\newblock \emph{J. Mach. Learn. Res.}, 21\penalty0 (89):\penalty0 1--53, 2020.

\bibitem[Hystad et~al.(2013)Hystad, Carpiano, Demers, Johnson, and
  Brauer]{hystad2013neighbourhood}
P.~Hystad, R.~M. Carpiano, P.~A. Demers, K.~C. Johnson, and M.~Brauer.
\newblock Neighbourhood socioeconomic status and individual lung cancer risk:
  evaluating long-term exposure measures and mediating mechanisms.
\newblock \emph{Social science \& medicine}, 97:\penalty0 95--103, 2013.

\bibitem[Ibfelt et~al.(2013)Ibfelt, Kjaer, H{\o}gdall, Steding-Jessen, Kjaer,
  Osler, Johansen, Frederiksen, and Dalton]{ibfelt2013socioeconomic}
E.~Ibfelt, S.~Kjaer, C.~H{\o}gdall, M.~Steding-Jessen, T.~Kjaer, M.~Osler,
  C.~Johansen, K.~Frederiksen, and S.~Dalton.
\newblock 'socioeconomic position and survival after cervical cancer: influence
  of cancer stage, comorbidity and smoking among danish women diagnosed between
  2005 and 2010.
\newblock \emph{British journal of cancer}, 109\penalty0 (9):\penalty0
  2489--2495, 2013.

\bibitem[Jackson(2018)]{jackson2018interpretation}
J.~W. Jackson.
\newblock On the interpretation of path-specific effects in health disparities
  research.
\newblock \emph{Epidemiology}, 29\penalty0 (4):\penalty0 517--520, 2018.

\bibitem[Jackson(2020)]{jackson2020meaningful}
J.~W. Jackson.
\newblock Meaningful causal decompositions in health equity research:
  definition, identification, and estimation through a weighting framework.
\newblock \emph{Epidemiology}, 32\penalty0 (2):\penalty0 282--290, 2020.

\bibitem[Jackson and VanderWeele(2018)]{jackson2018decomposition}
J.~W. Jackson and T.~J. VanderWeele.
\newblock Decomposition analysis to identify intervention targets for reducing
  disparities.
\newblock \emph{Epidemiology (Cambridge, Mass.)}, 29\penalty0 (6):\penalty0
  825, 2018.

\bibitem[Kasy(2014)]{kasy2014instrumental}
M.~Kasy.
\newblock Instrumental variables with unrestricted heterogeneity and continuous
  treatment.
\newblock \emph{The Review of Economic Studies}, 81\penalty0 (4):\penalty0
  1614--1636, 2014.

\bibitem[McGuire et~al.(2006)McGuire, Alegria, Cook, Wells, and
  Zaslavsky]{mcguire2006implementing}
T.~G. McGuire, M.~Alegria, B.~L. Cook, K.~B. Wells, and A.~M. Zaslavsky.
\newblock Implementing the institute of medicine definition of disparities: an
  application to mental health care.
\newblock \emph{Health services research}, 41\penalty0 (5):\penalty0
  1979--2005, 2006.

\bibitem[Oaxaca(1973)]{oaxaca1973male}
R.~Oaxaca.
\newblock Male-female wage differentials in urban labor markets.
\newblock \emph{International economic review}, pages 693--709, 1973.

\bibitem[PEARL(2001)]{pearl2001direct}
J.~PEARL.
\newblock Direct and indirect effects.
\newblock In \emph{Proceedings of the Seventeenth Conference on Uncertainty and
  Artificial Intelligence, 2001}, pages 411--420. Morgan Kaufman, 2001.

\bibitem[Pearl(2009)]{pearl2009causality}
J.~Pearl.
\newblock \emph{Causality}.
\newblock Cambridge university press, 2009.

\bibitem[Pearson et~al.(2012)Pearson, Lurie, Tosteson, Zhao, Abdu, and
  Weinstein]{pearson2012should}
A.~Pearson, J.~Lurie, T.~Tosteson, W.~Zhao, W.~Abdu, and J.~Weinstein.
\newblock Who should have surgery for spinal stenosis?: Treatment effect
  predictors in sport.
\newblock \emph{Spine}, 37\penalty0 (21):\penalty0 1791, 2012.

\bibitem[Quirk et~al.(2006)Quirk, Harbors, and Design]{quirk2006administrative}
J.~Quirk, B.~Harbors, and P.~Design.
\newblock Administrative draft existing conditions report.
\newblock 2006.

\bibitem[Robins and Greenland(1992)]{robins1992identifiability}
J.~M. Robins and S.~Greenland.
\newblock Identifiability and exchangeability for direct and indirect effects.
\newblock \emph{Epidemiology}, pages 143--155, 1992.

\bibitem[Sch{\"o}lkopf et~al.(2021)Sch{\"o}lkopf, Locatello, Bauer, Ke,
  Kalchbrenner, Goyal, and Bengio]{scholkopf2021toward}
B.~Sch{\"o}lkopf, F.~Locatello, S.~Bauer, N.~R. Ke, N.~Kalchbrenner, A.~Goyal,
  and Y.~Bengio.
\newblock Toward causal representation learning.
\newblock \emph{Proceedings of the IEEE}, 109\penalty0 (5):\penalty0 612--634,
  2021.

\bibitem[Spirtes and Richardson(1996)]{spirtes1996polynomial}
P.~Spirtes and T.~Richardson.
\newblock A polynomial time algorithm for determining dag equivalence in the
  presence of latent variables and selection bias.
\newblock In \emph{Proceedings of the 6th International Workshop on Artificial
  Intelligence and Statistics}, pages 489--500, 1996.

\bibitem[Spirtes et~al.(2000)Spirtes, Glymour, Scheines, and
  Heckerman]{spirtes2000causation}
P.~Spirtes, C.~N. Glymour, R.~Scheines, and D.~Heckerman.
\newblock \emph{Causation, prediction, and search}.
\newblock MIT press, 2000.

\bibitem[Valeri and VanderWeele(2013)]{valeri2013mediation}
L.~Valeri and T.~J. VanderWeele.
\newblock Mediation analysis allowing for exposure--mediator interactions and
  causal interpretation: theoretical assumptions and implementation with sas
  and spss macros.
\newblock \emph{Psychological methods}, 18\penalty0 (2):\penalty0 137, 2013.

\bibitem[VanderWeele(2011)]{vanderweele2011controlled}
T.~J. VanderWeele.
\newblock Controlled direct and mediated effects: definition, identification
  and bounds.
\newblock \emph{Scandinavian Journal of Statistics}, 38\penalty0 (3):\penalty0
  551--563, 2011.

\bibitem[VanderWeele et~al.(2014)VanderWeele, Vansteelandt, and
  Robins]{vanderweele2014effect}
T.~J. VanderWeele, S.~Vansteelandt, and J.~M. Robins.
\newblock Effect decomposition in the presence of an exposure-induced
  mediator-outcome confounder.
\newblock \emph{Epidemiology (Cambridge, Mass.)}, 25\penalty0 (2):\penalty0
  300, 2014.

\bibitem[Vansteelandt and VanderWeele(2012)]{vansteelandt2012natural}
S.~Vansteelandt and T.~J. VanderWeele.
\newblock Natural direct and indirect effects on the exposed: effect
  decomposition under weaker assumptions.
\newblock \emph{Biometrics}, 68\penalty0 (4):\penalty0 1019--1027, 2012.

\bibitem[Vos et~al.(2017)Vos, Abajobir, Abate, Abbafati, Abbas, Abd-Allah,
  Abdulkader, Abdulle, Abebo, Abera, et~al.]{vos2017global}
T.~Vos, A.~A. Abajobir, K.~H. Abate, C.~Abbafati, K.~M. Abbas, F.~Abd-Allah,
  R.~S. Abdulkader, A.~M. Abdulle, T.~A. Abebo, S.~F. Abera, et~al.
\newblock Global, regional, and national incidence, prevalence, and years lived
  with disability for 328 diseases and injuries for 195 countries, 1990--2016:
  a systematic analysis for the global burden of disease study 2016.
\newblock \emph{The Lancet}, 390\penalty0 (10100):\penalty0 1211--1259, 2017.

\bibitem[Weinstein et~al.(2006{\natexlab{a}})Weinstein, Lurie, Tosteson,
  Skinner, Hanscom, Tosteson, Herkowitz, Fischgrund, Cammisa, Albert,
  et~al.]{weinstein2006surgical2}
J.~N. Weinstein, J.~D. Lurie, T.~D. Tosteson, J.~S. Skinner, B.~Hanscom, A.~N.
  Tosteson, H.~Herkowitz, J.~Fischgrund, F.~P. Cammisa, T.~Albert, et~al.
\newblock Surgical vs nonoperative treatment for lumbar disk herniation: the
  spine patient outcomes research trial (sport) observational cohort.
\newblock \emph{Jama}, 296\penalty0 (20):\penalty0 2451--2459,
  2006{\natexlab{a}}.

\bibitem[Weinstein et~al.(2006{\natexlab{b}})Weinstein, Tosteson, Lurie,
  Tosteson, Hanscom, Skinner, Abdu, Hilibrand, Boden, and
  Deyo]{weinstein2006surgical1}
J.~N. Weinstein, T.~D. Tosteson, J.~D. Lurie, A.~N. Tosteson, B.~Hanscom, J.~S.
  Skinner, W.~A. Abdu, A.~S. Hilibrand, S.~D. Boden, and R.~A. Deyo.
\newblock Surgical vs nonoperative treatment for lumbar disk herniation: the
  spine patient outcomes research trial (sport): a randomized trial.
\newblock \emph{Jama}, 296\penalty0 (20):\penalty0 2441--2450,
  2006{\natexlab{b}}.

\bibitem[Weinstein et~al.(2008)Weinstein, Tosteson, Lurie, Tosteson, Blood,
  Hanscom, Herkowitz, Cammisa, Albert, Boden, et~al.]{weinstein2008surgical}
J.~N. Weinstein, T.~D. Tosteson, J.~D. Lurie, A.~N. Tosteson, E.~Blood,
  B.~Hanscom, H.~Herkowitz, F.~Cammisa, T.~Albert, S.~D. Boden, et~al.
\newblock Surgical versus nonsurgical therapy for lumbar spinal stenosis.
\newblock \emph{New England Journal of Medicine}, 358\penalty0 (8):\penalty0
  794--810, 2008.

\bibitem[Zhang(2008{\natexlab{a}})]{ZHANG20081873}
J.~Zhang.
\newblock On the completeness of orientation rules for causal discovery in the
  presence of latent confounders and selection bias.
\newblock \emph{Artificial Intelligence}, 172\penalty0 (16):\penalty0
  1873--1896, 2008{\natexlab{a}}.
\newblock ISSN 0004-3702.
\newblock \doi{https://doi.org/10.1016/j.artint.2008.08.001}.
\newblock URL
  \url{https://www.sciencedirect.com/science/article/pii/S0004370208001008}.

\bibitem[Zhang(2008{\natexlab{b}})]{zhang2008causal}
J.~Zhang.
\newblock Causal reasoning with ancestral graphs.
\newblock \emph{Journal of Machine Learning Research}, 9:\penalty0 1437--1474,
  2008{\natexlab{b}}.

\bibitem[Zhang and Spirtes(2005)]{zhang2005transformational}
J.~Zhang and P.~Spirtes.
\newblock A transformational characterization of markov equivalence for
  directed acyclic graphs with latent variables.
\newblock In \emph{Proceedings of the Twenty-First Conference on Uncertainty in
  Artificial Intelligence}, pages 667--674, 2005.

\end{thebibliography}

\appendix
\onecolumn 

\section{Theoretical Analysis}
\label{app:thm}

Our theoretical analysis is based on Prop. 3.1 and some MAG-related definitions. These definitions can be found in \cite{zhang2008causal}. We will also introduce them here for completeness.

We begin with the \emph{mixed graph}, which is a graph with two kinds of edges: directed ($\to$), and bi-directed ($\leftrightarrow$). Note that this definition is slightly different from that in \cite{zhang2008causal} as we do not require it to contain an undirected graph ($-$), which means two variables are dependent due to latent selection variables. According to our problem setting, the observed variables are not assumed to be dependent on some unobserved selection variables, so it is unnecessary for our definition to contain the undirected graph. We first introduce the definition of \emph{ancestral graph}.

\begin{definition}[Ancestral Graph]
An ancestral graph $G$ is a mixed graph with no directed and almost directed cycles. Here, the almost directed cycle occurs when $A \leftrightarrow B$ and $B \in \mathrm{An}_{G}(A)$. 
\end{definition}

Next, we introduce the definition of inducing path, which can further define \emph{maximal ancestral graph}.

\begin{definition}[Inducing Path]
A path $p$ from $X$ to $Y$ in $G$ is an inducing path if every non-endpoint vertex on $p$ is a collider and is either $\in \mathrm{An}_{G}(X)$ or $\in \mathrm{An}_{G}(Y)$.
\end{definition}

\begin{definition}[Maximal Ancestral Graph]
An ancestral graph is Maximal Ancestral Graph if, for any two non-adjacent variables, there is no inducing path between them. 
\end{definition}

Next, we introduce the definition of \emph{unshielded collider} and \emph{discriminating path}, which form two main principles in Prop.~3.1. 

\begin{definition}[Unshielded Collider]
A triple $\langle X,Y,Z \rangle$ is called an unshielded collider if the edge between $(X,Y)$ and the one between $(Y,Z)$ are pointed to $Y$. 
\end{definition}

\begin{definition}[Discriminating path]
In a MAG, a path $p=\langle X,V_1,...,V_p, W,Y \rangle$, is a discriminating path for $W$ if \textbf{i)} $p$ includes at least three edges; \textbf{ii)} $W$ is a non-endpoint on $p$, and is adjacent to $Y$. \textbf{iii)} $X$ is not adjacent to $Y$, and each $V_i$ for $i \in \{1,...,p\}$ is a collider on $p$ and $\in \mathrm{Pa}(Y)$.
\end{definition}

\begin{proof}[Proof of Prop. 5.1]

The \cite{galles2013testing} provided four conditions for $p(y|do(m),r)$ to be identifiable, and the theorem 4.3.2 in \cite{pearl2009causality} proved that at least one of these conditions holds. These four conditions are:
\begin{itemize}
    \item There is no-back door path from $M$ to $Y$.
    \item There is no directed path from $M$ to $Y$ in $G$. 
    \item There exists a set of nodes $\mathbf{B}$ that blocks all back-door paths from $M$ to $Y$. 
    \item There exists sets of nodes $Z_1$ and $Z_2$ such that:
    \begin{enumerate}
        \item $Z_1$ blocks every directed path from $M$ to $Y$;
        \item $Z_2$ blocks all back-door paths between $Z_1$ and $Y$;
        \item $Z_2$ blocks all back-door paths between $M$ and $Z_1$;
        \item $Z_2$ does not activate any back-door paths between $M$ and $Y$. 
    \end{enumerate}
\end{itemize}

Since there is a directed edge $M \to Y$, therefore the second and the fourth conditions are violated. As we condition on $R$, so the first condition can be modified to no back-door path from $M$ to $Y$, given $R$. In this regard, the $\mathbf{B}$ in the third condition degenerates to $\emptyset$. So we have $p(y|do(m),r)$ is identifiable if and only if there exists $\mathbf{B}$ that block all back-door paths from $M$ to $Y$. If this exists unobserved confounder between $M$ to $Y$, this sufficient and necessary condition will not hold, making $p(y|do(m),r)$ unidentifiable. 
\end{proof}

\begin{proof}[Proof of Theorem. 5.1]
It is sufficient to show that the \textbf{i)} $\underline{\mathbf{B}}_M$ contain all mediator-outcome confounders and meanwhile \textbf{ii)} the variable in $\overline{\mathbf{B}}_M$ does not perturb existing directed paths from $M$ to $Y$ and not induce new spurious features, according to \cite{pearl2009causality}. For \textbf{i)}, note that for each mediator-outcome confounder, it causally influences both $M$ and $Y$, which means this there exists directed paths from $O_i$ to $M$ and $O_i$ to $Y$. Therefore, the maximal ancestral graph over $(M,O_i,Y)$ is $\{M \gets O_i \to Y\} \cup \{M \to Y\}$. Therefore, we have $O_i \in \underline{\mathbf{B}}_M$. For \textbf{ii)}, for each $O_i \in \overline{\mathbf{B}}_M$, if it lies in the directed path from $M$ to $Y$, then the learned MAG must be $\{M \to O_i \to Y\} \cup \{M \to Y\}$, which is not contained in $\overline{\mathbf{B}}_M$. For $O_i$ with the MAG (in addition to the edge $M \to Y$) with$M \gets O_i \leftrightarrow Y$ or $M \leftrightarrow O_i \to Y$. Since there is no unobserved confounder between $M$ and $Y$, then there respectively exists an observed confounder that blocks the path between $O_i$ and $Y$ and the path between $M$ and $O_i$. Besides, such a confounder is in $\underline{\mathbf{B}}_M$. Therefore, such an $O_i$ will not induce new spurious features, since even there is a collider induces this path, there exists a variable in $\underline{\mathbf{B}}_M$ that can block this path. In this regard, any set that contains $\underline{\mathbf{B}}_M$ and belong to $\overline{\mathbf{B}}_M$ will satisfy both \textbf{i)} and \textbf{ii)}. Therefore, this set is admissible. 
\end{proof}

\begin{proof}[Proof of Lemma~5.1]
It is easy to note that we have $O_i \in \overline{\mathbf{B}}_M$ if we can discriminate the associated types of MAGs from those in $\tilde{\mathbf{B}}_{\mathrm{MAG}}(M)$. Specifically, the MAGs in $\tilde{\mathbf{B}}_{\mathrm{MAG}}(M)$ (we omit $M \to Y$ as it is included in all MAGs) are: i) $M \gets O_i \to Y$; ii) $M \gets O_i \leftrightarrow Y$; iii) $M \leftrightarrow O_i \to Y$; iv) $M \to O_i \to Y$; v) $M \to O_i \leftrightarrow Y$; vi) $M \leftrightarrow O_i \leftrightarrow Y$. Therefore, our goal is to discriminate i), ii), iii) from others. Note that when the first condition in Lemma~5.1 is satisfied, we can discriminate i) and ii) using unshielded collider condition in Prop.~3.1, as the $M \to O_i$ or $M \leftrightarrow O_i$ can make $\langle O_i, M, C_M^i \rangle$ unshielded collier, therefore we have $O_i \in \overline{\mathbf{B}}_M$. If the second condition is satisfied, we can first discriminate i), ii), iii), vi) from iv) and v) using unshielded collider condition; then to eliminate vi) which can induce spurious path for instance $M \gets S_1 \to O_i \gets S_2 \to Y$ with $S_1, S_2$ unobserved, we note that $\langle C_M^j, M, O_i, Y \rangle$ form a discriminating path for $O_i$ in i), ii), iii) and vi), with the difference that the $O_i$ in vi) is a collider on the path $M \leftrightarrow O_i \leftrightarrow Y$ in vi) but is not a collider on the same path in i), ii) and iii). Therefore, we can discriminate i), ii), iii) from vi). In this regard, we can determine whether $O_i \in \overline{\mathbf{B}}_M$ if either one of the condition holds in Lemma 5.1. 
\end{proof}

\begin{proof}[Proof of Lemma 5.2]
For the $O_i$ such that $M \gets O_i \to Y$ in the MAG over $(O_i,M,Y)$, we have that the deconfounding sets between $O_i$ and $M$, and that between $M$ and $Y$, can be observed, thus we there exists $\mathbf{X}_i$ which contains both deconfounding sets, such that $\{P^r(M|O_i)\}$, $\{P^r(Y|O_i)\}$ are independent to $\{P^r(O_i)\}$, because they have disentangled causal mechanisms. Besides, we have $\{P^r(O_i|M)\}$, $\{P^r(O_i|Y)\}$ are respectively dependent to $\{P^r(M)\}$ and $\{P^r(Y)\}$ under assumption~5.3, since they are entangled. On the other hand, we will show that this conclusion does not hold for $O_i$ with other types of MAG. Similarly, the MAGs in $\tilde{\mathbf{B}}_{\mathrm{MAG}}(M)$ (we omit $M \to Y$ as it is included in all MAGs) are: i) $M \gets O_i \to Y$; ii) $M \gets O_i \leftrightarrow Y$; iii) $M \leftrightarrow O_i \to Y$; iv) $M \to O_i \to Y$; v) $M \to O_i \leftrightarrow Y$; vi) $M \leftrightarrow O_i \leftrightarrow Y$. For the MAG with $M \leftrightarrow O_i$, then for any set $\mathbf{A}$, the independent relations for  ($\{P^r(M|\mathbf{A},O_i)\}$,$\{P^r(O_i|\mathbf{A})\}$) and ($\{P^r(M|\mathbf{A})\}$,$\{P^r(O_i|M,\mathbf{A})\}$) simultaneously hold or not hold, determined by whether $A$ is the deconfounding set between $O_i$ and $M$. This conclusion similarly holds for the MAG with $O_i \leftrightarrow Y$. Therefore, we can discriminate MAGs ii), iii),v), vi) from others. For the MAG with $M \to O_i \to Y$, if there exists a set $\mathbf{A}$ such that $\{P^r(M|\mathbf{A},O_i)\}$ is $\{P^r(O_i|\mathbf{A})\}$ independent, then $\mathbf{A}$ must block the directed paths between $M$ and $O_i$. Suppose otherwise there exists an unblocked directed path from $M$ to $O_i$; in this regard, both distributions are entangled since they are affected by the mechanism $P^r(O_i|\mathbf{A},M)$. Then if $\mathbf{A}$ blocks the directed paths between $M$ and $O_i$, the independence between $\{P^r(O_i|\mathbf{A},M)\}$ and $\{P^r(M|\mathbf{A})\}$ also holds. This can help us further eliminate type iv). As a result, we have $O_i \in \underline{\mathbf{B}}_M$. 
\end{proof}

\begin{proof}[Proof of Theorem 5.2]
Under assumptions 5.1, 5.2, and 5.3, we can detect $\underline{\mathbf{B}}_M$. Besides, if condition in Lemma~5.1 holds for $O_i$, we can determine whether $O_i \in \overline{\mathbf{B}}_M$. Therefore, we have the detected $\mathbf{B}_M$ is 
$$
\underline{\mathbf{B}}_M \cup \{O_i|\text{the condition in Lemma~5.1 for $O_i$ holds and } O_i \in \overline{\mathbf{B}}_M\}, which \supset \underline{\mathbf{B}}_M and \subset \overline{\mathbf{B}}_M.
$$
The proof is completed. 
\end{proof}

\section{Experiment on SPORT}
\label{app:exp}

\subsection{Implementation Details}

Our implementation is composed of three steps: \textbf{i)} detect the mediator set $\mathbf{M}$; \textbf{ii)} identify admissible sets; and \textbf{iii)} Estimation.

\textbf{Detection of  $\mathbf{M}$}. We follow the FCI-JCI algorithm \cite{ZHANG20081873} to identify adjustment sets among social-economic factors. In total, $M$ contains four features: income and working styles (social economics status) Body mass index (BMI), and the outcome score at baseline (health status).

\textbf{Identification of admissible sets.} In addition to the race variable, there are also other context variables such as gender, age, \emph{etc.} We discovered that there is at least a context variable that is non-adjacent to $M$ for each adjusted factor $M$. As a result, we can directly identify $\mathbf{B}_{M}$ by iteratively implementing the FCI-JCI algorithm to obtain MAG over $(R,M,O,Y)$ for each $O \in \mathbf{B}^{\mathrm{all}}_M$. Specifically, the $\mathbf{B}_M$ is (age, gender, education years, working style) when $M$ denotes income; is (age, gender, education years, back pain at baseline $(bp_{BA})$) when $M$ denotes working style; (gender, diabetes) when $M$ denotes BMI and finally is $\emptyset$ when $M:=bp_{BA}$. 

\textbf{Estimation.} To fit $\mathbb{E}(Y|\mathbf{B}_M,M,R)$, we train a random forest predictor with $50$ estimators.

\subsection{Results on SPS, DS diseases when $R$ denotes "Race"}
\label{app:exp-race}
In this section, we show our estimated adjusted and unadjusted effects of SPS and DS diseases in the SPORT dataset. Fig.~\ref{fig:sps} and Fig.~\ref{fig:ds} respectively show the adjusted and unadjusted effects of the SPS and DS diseases. Fig.~\ref{fig:total} shows the total effects.

\begin{figure}[ht!]
    \centering
    \includegraphics[width=5.5in]{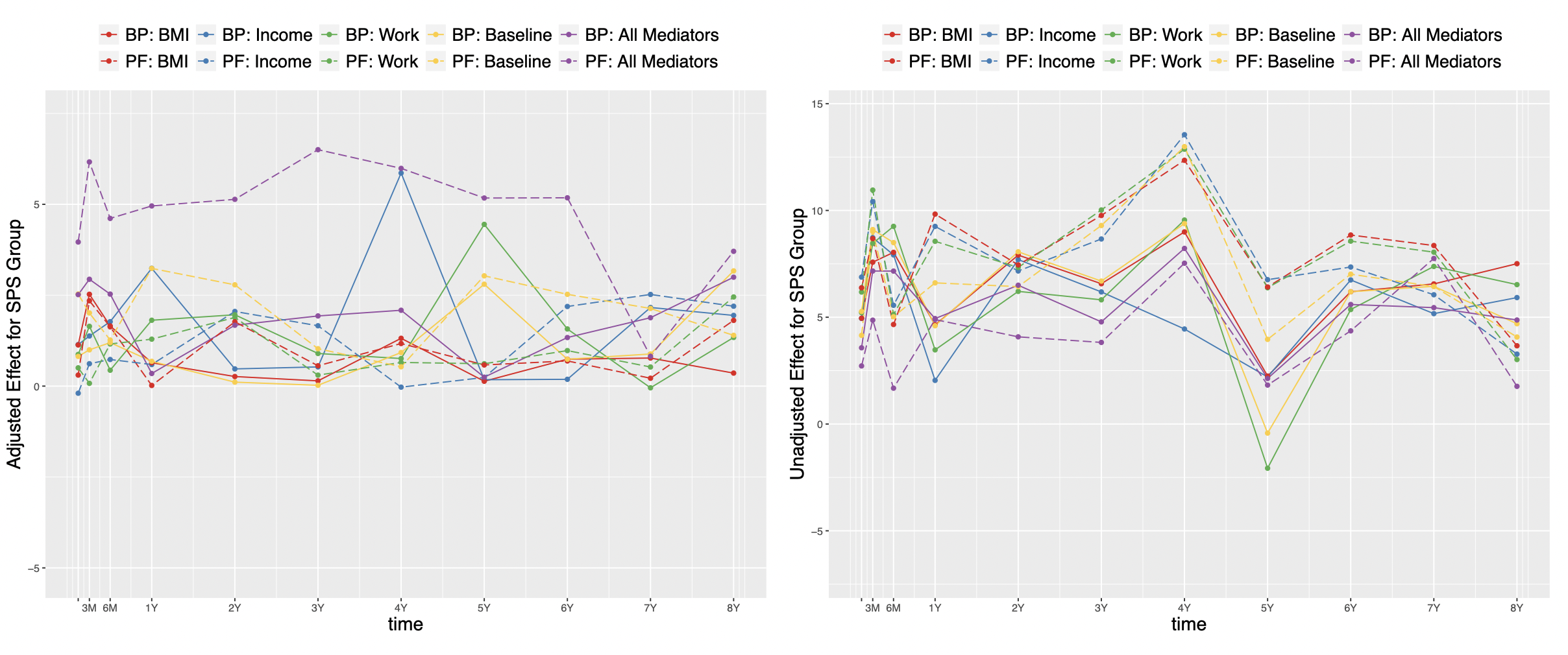}
    \caption{Adjusted and Unadjusted effects of the SPS disease when $R$ denotes race.}
    \label{fig:sps}
\end{figure}

\begin{figure}[ht!]
    \centering
    \includegraphics[width=5.5in]{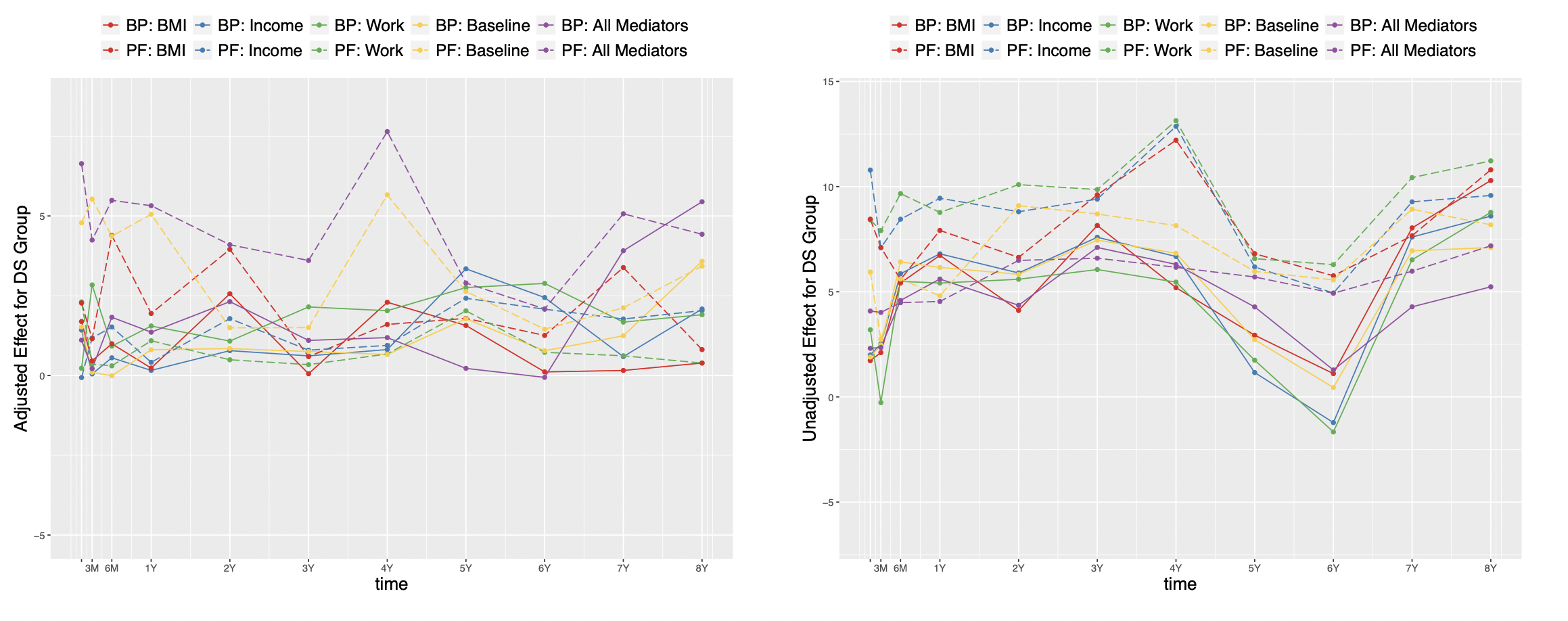}
    \caption{Adjusted and Unadjusted effects of the DS disease when $R$ denotes race.}
    \label{fig:ds}
\end{figure}

\begin{figure}[ht!]
    \centering
    \includegraphics[width=5.5in]{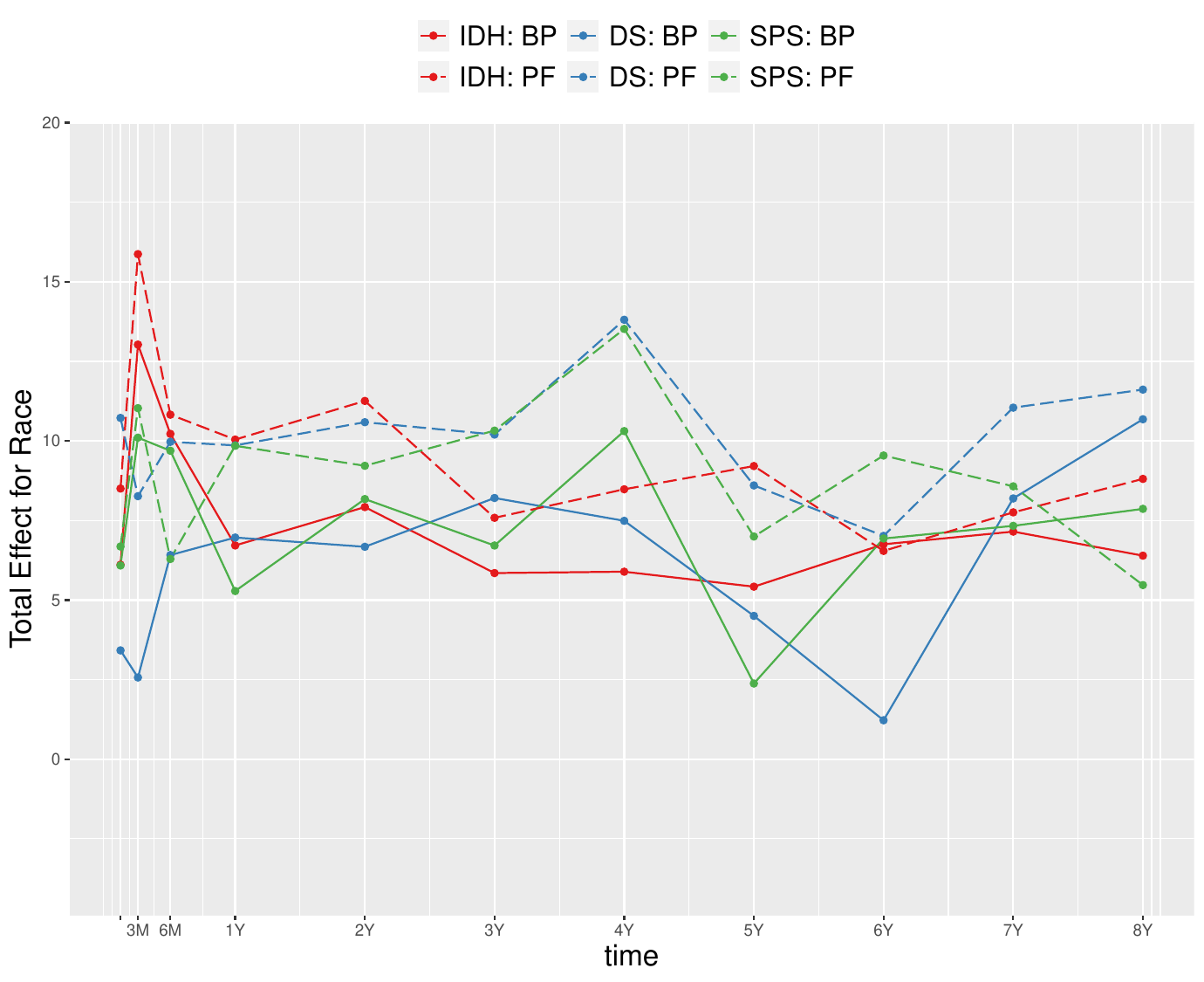}
    \caption{Total Effects with $Y$ denoting BF and PF scores when $R$ denotes race.}
    \label{fig:total}
\end{figure}

\newpage
\subsection{Results on when $R$ denotes "Gender"}
\label{app:exp-gender}
We also conduct experiments when $R$ denotes "Gender", with $R=1$ (or $R=0$) meaning the male (female) group. Fig.~\ref{fig:total} shows the total effects of $\mathbb{E}(Y|R=1)-\mathbb{E}(Y|R=0)$, which suggests a significant difference in SPS and DS diseases; and a minor but unignorable difference in IDH disease. The selected mutable variable set $\mathbf{M}$ contains "income", "working styles" and "BMI". The $\mathbf{B}_M$ is (age, gender, education years, working style) when $M$ denotes income; is (age, gender, education years, back pain at baseline $(bp_{BA})$) when $M$ denotes working style; (gender, diabetes) when $M$ denotes BMI. The whole implementation procedure is the same to the one when $R$ denotes "Race".  Fig.~\ref{fig:idh-g},~\ref{fig:sps-g} and Fig.~\ref{fig:ds-g} respectively shows the adjusted and unadjusted effects of the IDH, SPS, and DS diseases.

\begin{figure}[ht!]
    \centering
    \includegraphics[width=5.5in]{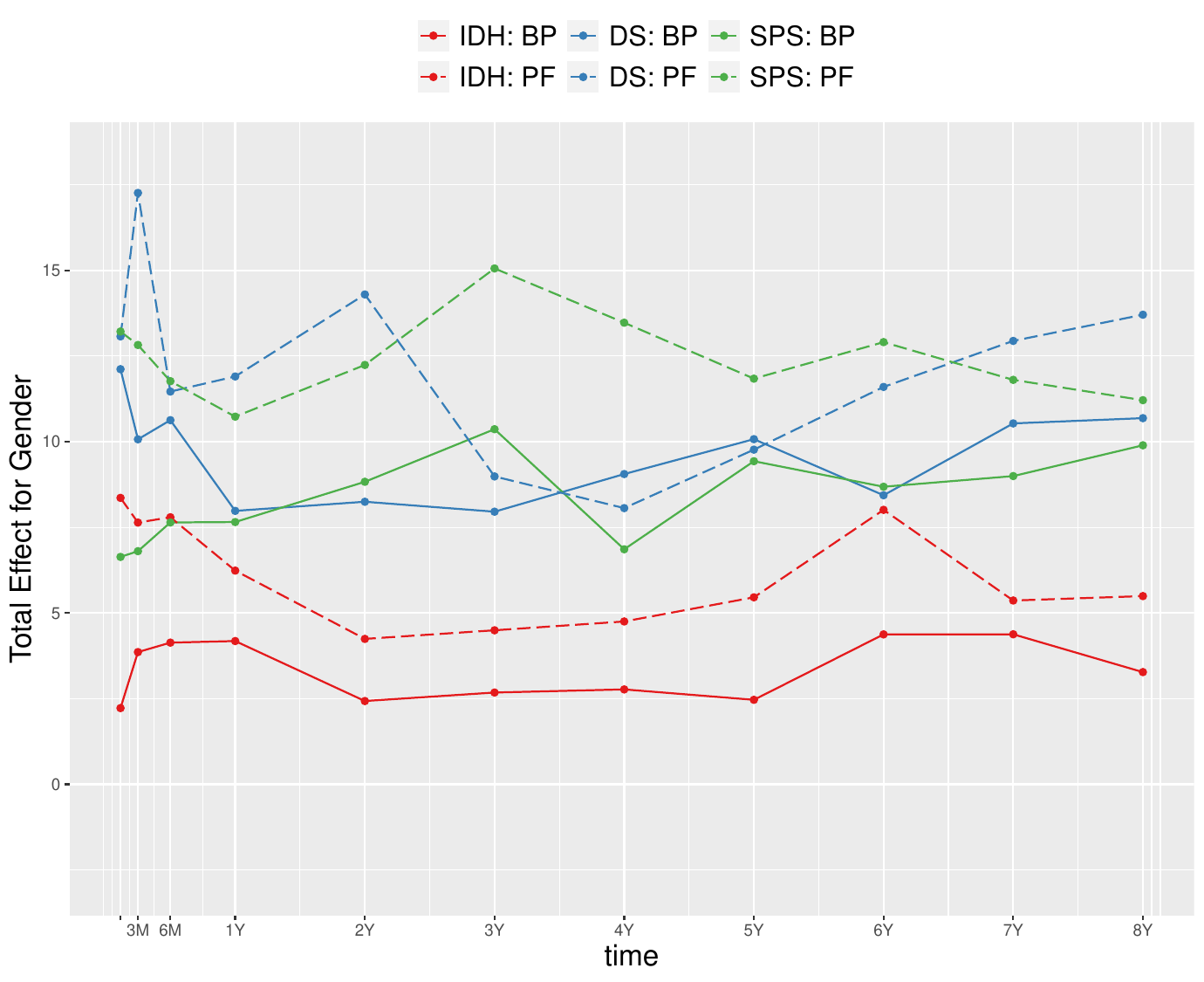}
    \caption{Total Effects with $R$ denotes gender.}
    \label{fig:total-g}
\end{figure}

\begin{figure}[ht!]
    \centering
    \includegraphics[width=5.5in]{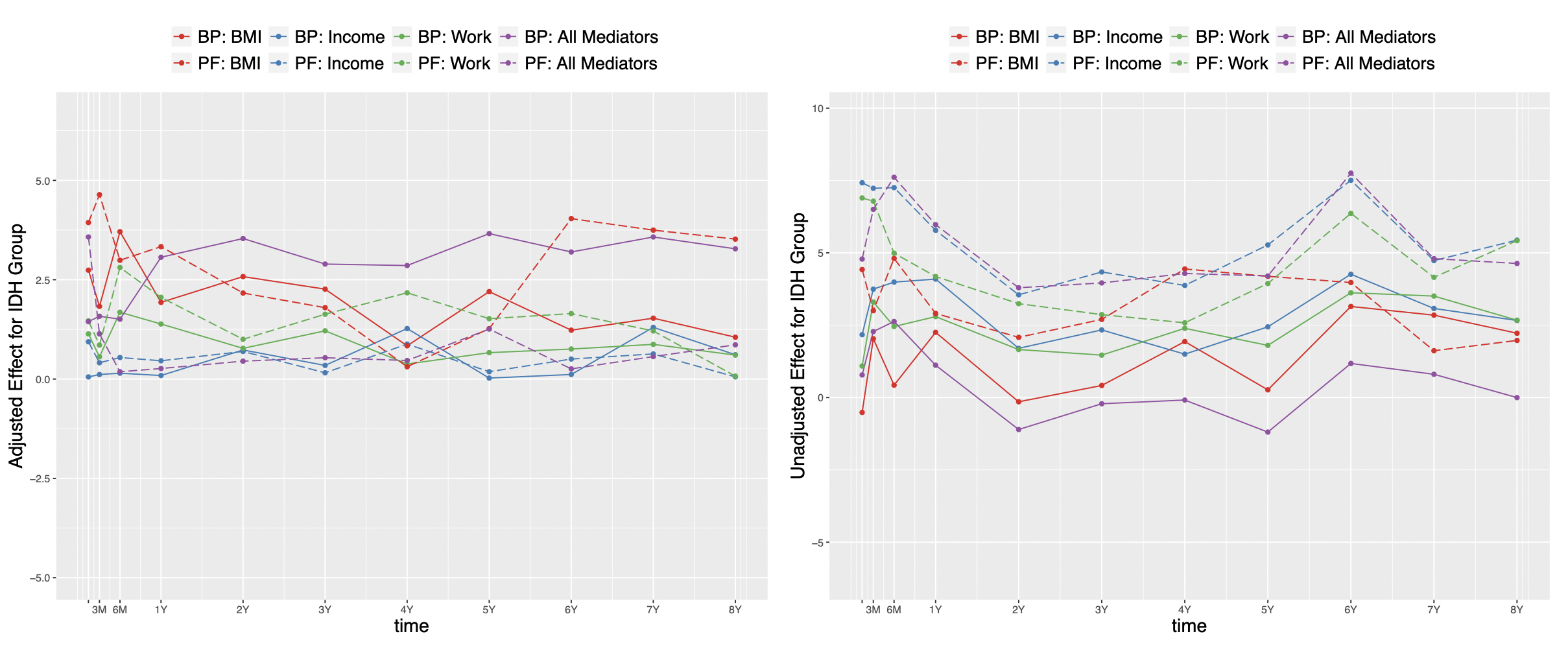}
    \caption{Adjusted and Unadjusted effects of the IDH diseases when $R$ denotes gender.}
    \label{fig:idh-g}
\end{figure}

\begin{figure}[ht!]
    \centering
    \includegraphics[width=5.5in]{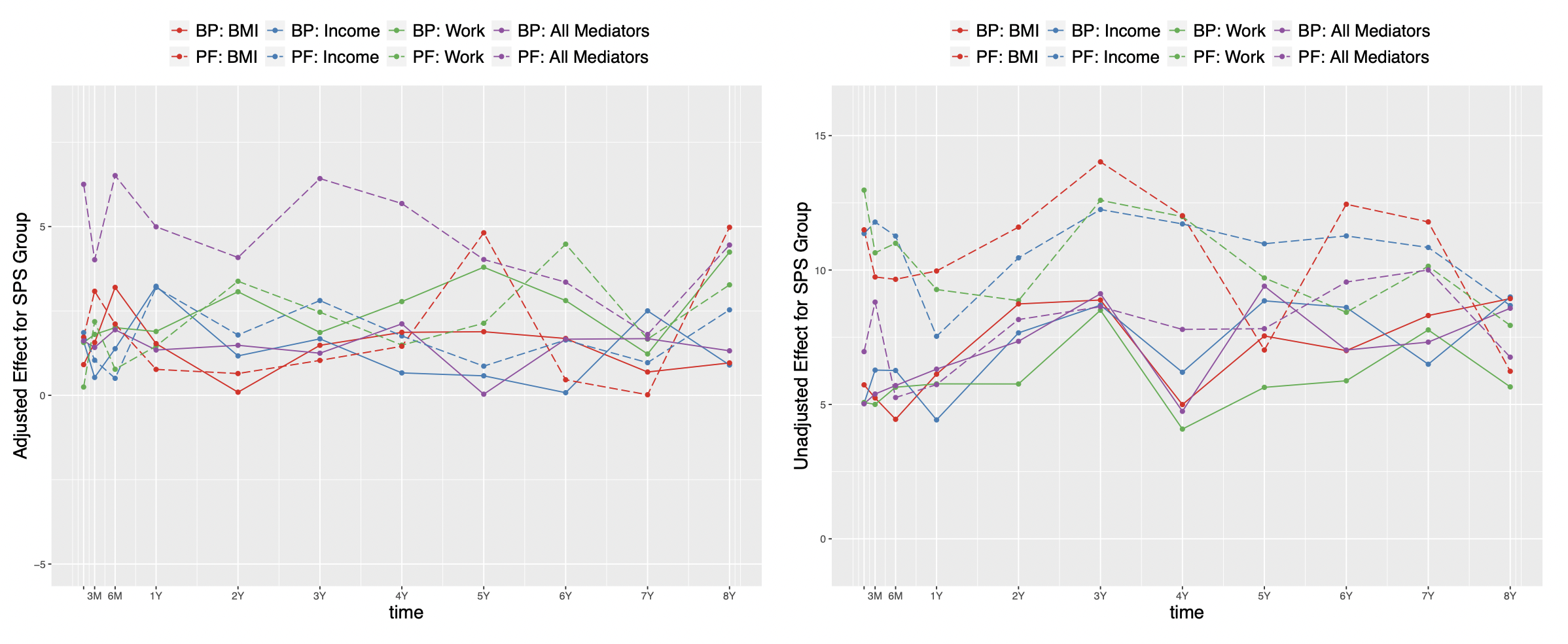}
    \caption{Adjusted and Unadjusted effects of the SPS disease when $R$ denotes gender.}
    \label{fig:sps-g}
\end{figure}

\begin{figure}[ht!]
    \centering
    \includegraphics[width=5.5in]{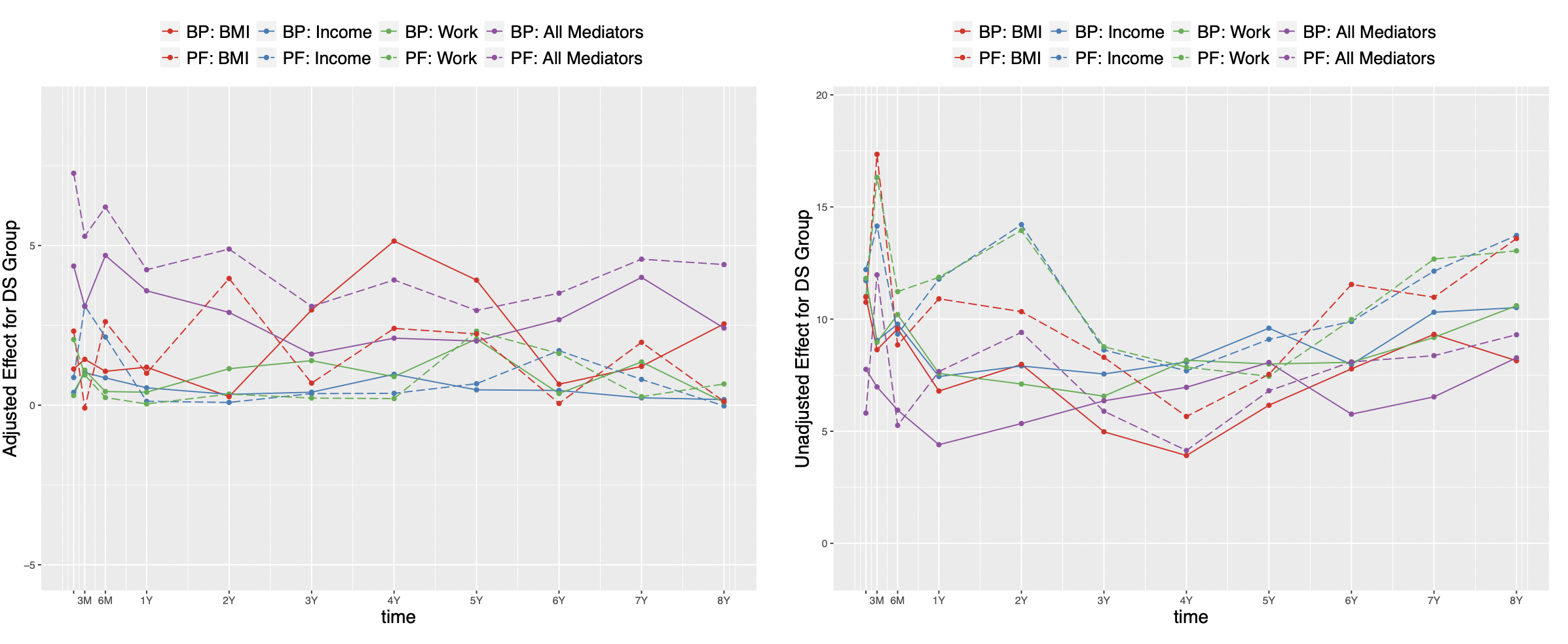}
    \caption{Adjusted and Unadjusted effects of the DS disease when $R$ denotes gender.}
    \label{fig:ds-g}
\end{figure}

\end{document}